\documentclass[aps,pra,reprint,twocolumn,groupedaddress,floatfix]{revtex4-1}

\usepackage{graphicx}
\usepackage{stackengine}
\usepackage{braket}
\usepackage{amsmath}
\usepackage{qcircuit}
\usepackage{amsthm}
\usepackage{hyperref}
\hypersetup{
    colorlinks=true,
    allcolors=blue
}
\newtheorem{innercustomgeneric}{\customgenericname}
\providecommand{\customgenericname}{}
\newcommand{\newcustomtheorem}[2]{%
  \newenvironment{#1}[1]
  {%
   \renewcommand\customgenericname{#2}%
   \renewcommand\theinnercustomgeneric{##1}%
   \innercustomgeneric
  }
  {\endinnercustomgeneric}
}
\newcustomtheorem{customthm}{Theorem}
\newcustomtheorem{customlemma}{Lemma}
\newtheorem{theorem}{Theorem}
\newtheorem{lemma}[theorem]{Lemma}
\newcommand{\comment}[1]{}

\DeclareMathOperator{\rank}{rank}

\begin{document}


\title{Three-dimensional surface codes: Transversal gates and fault-tolerant architectures}


\author{Michael Vasmer}
\email[]{michael.vasmer.15@ucl.ac.uk}
\author{Dan E. Browne}
\affiliation{Department of Physics and Astronomy, University College London, Gower Street, London WC1E 6BT, United Kingdom}


\date{\today}

\begin{abstract}
One of the leading quantum computing architectures is based on the two-dimensional (2D) surface code. This code has many advantageous properties such as a high error threshold and a planar layout of physical qubits where each physical qubit need only interact with its nearest neighbours. However, the transversal logical gates available in 2D surface codes are limited. This means that an additional (resource intensive) procedure known as magic state distillation is required to do universal quantum computing with 2D surface codes. Here, we examine three-dimensional (3D) surface codes in the context of quantum computation. We introduce a picture for visualizing 3D surface codes which is useful for analysing stacks of three 3D surface codes. We use this picture to prove that the $CZ$ and $CCZ$ gates are transversal in 3D surface codes. We also generalize the techniques of 2D surface code lattice surgery to 3D surface codes. We combine these results and propose two quantum computing architectures based on 3D surface codes. Magic state distillation is not required in either of our architectures. Finally, we show that a stack of three 3D surface codes can be transformed into a single 3D color code (another type of quantum error-correcting code) using code concatenation.  
\end{abstract}


\maketitle


\section{\label{sec:intro} Introduction}

The family of quantum error-correcting codes known as surface codes (also called toric codes or homological codes) have generated a great deal of theoretical and experimental interest since their introduction by Kitaev~\cite{Kitaev2003Anyons}. We can define surface codes in any spatial dimension $D\geq 2$. The two-dimensional (2D) surface code~\cite{BravyiKitaev1998Boundary,Dennis2002Topological} is the basis of one of the leading proposals for a fault-tolerant quantum computing architecture~\cite{Fowler2012Surface}. The biggest advantage of the 2D surface code is its high error threshold which approaches $1\%$~\cite{Fowler2009High,Raussendorf2007Fault,Stephens2014Fault}, a value which has been achieved in various qubit technologies~\cite{Barends2014Logic,Harty2014HighFidelity}. The other main advantage of the 2D surface code is that it has a simple structure consisting of a planar layout of qubits where each qubit only needs to interact with four neighbouring qubits. Experimental groups in universities and industry are targeting the surface code as their eventual fault-tolerant architecture~\cite{Barends2014Logic,Vijay2015Majorana,Hill2015Surface,OGorman2016Silicon,Lekitsch2017Blueprint}. However, these groups are still a long way off the millions of qubits required to run quantum algorithms such as Shor's algorithm~\cite{Shor1999Factoring} on a surface code quantum computer~\cite{Fowler2012Surface,OGorman2017MagicFactories}. One of the contributing factors to the large qubit overhead of 2D surface code architectures is that a procedure known as magic state distillation is needed if we want to implement the non-Clifford $T$ gate~\cite{Bravyi2005Universal} in 2D surface codes. Non-Clifford gates are required for universal quantum computation but they are rarely easy to implement in quantum error-correcting codes. Magic state distillation is estimated to have a resource cost $\sim 150--300$ times greater than the resource cost of realizing the control-NOT ($CNOT$) gate in 2D surface code architectures~\cite{OGorman2017MagicFactories}. The overhead associated with magic state distillation has motivated research into alternative methods for realizing non-Clifford gates in topological codes. For example, the 3D gauge color code has a transversal non-Clifford gate and this code forms the basis of a universal quantum computing architecture with attractive properties such as single-shot error correction~\cite{Bombin2015GaugeCC,Bombin2015Single,Brown2016Gauge}. The resource overheads of 3D gauge color codes and 2D surface codes with magic state distillation are estimated to scale in a similar way~\cite{OGorman2017MagicFactories}, so for different ranges of parameters either option could be advantageous. However, it has been argued that 2D surface code architectures will be superior for current experimental parameters due to the superlative error threshold of 2D surface codes~\cite{OGorman2017MagicFactories}.

In this article, we study three-dimensional (3D) surface codes. These codes were first introduced in~\cite{Dennis2002Topological} and their topological entropy was studied in~\cite{CastelnovoChamon2008Topological3DToric}. Most previous work on 3D surface codes in the context of quantum computing has concentrated on the relationship between 3D surface codes and 3D color codes. Color codes are another family of topological error-correcting codes which share some features with surface codes. It turns out that we can transform any 3D color code into three 3D surface codes using local Clifford unitaries~\cite{Kubica2015Unfolding}. This relationship has implications for quantum computing with 3D surface codes and 3D color codes. For example, using the mapping between the two code families, we can use 3D surface code decoders to decode 3D color codes~\cite{Aloshious2016Projecting}. This is useful because efficient 3D color code decoders are difficult to construct. Color codes tend to have a larger range of transversal logical gates when compared with surface codes~\cite{Bombin2007Topological,KubicaBeverland2015Colour}. This implies that we can use the relationship between surface codes and color codes to realize logical gates in 3D surface codes which are not naively available. Indeed, we can use the mapping between the two code families to implement a locality-preserving control-control-$Z$ ($CCZ$) gate in the 3D surface code~\cite{Kubica2015Unfolding}. Locality-preserving logical operators (LPLOs) are naturally fault-tolerant because the growth of errors under a LPLO is bounded by a constant~\cite{Bravyi2013Class,Webster2018Locality}. Recently, the LPLOs of 3D surface codes with different boundary conditions were classified using a correspondence between logical operators and domain walls~\cite{Webster2018Locality}. 

Here, we introduce a way of visualizing 3D surface codes, which we call the rectified picture. We use the rectified picture to analyse stacks of three 3D surface codes. We show that 3D surface codes possess more transversal gates than was previously thought, namely that both the control-$Z$ ($CZ$) and $CCZ$ gates are transversal for stacks of three 3D surface codes. These results build on the results in~\cite{Kubica2015Unfolding} and \cite{Webster2018Locality}, where it was shown that $CZ$ and $CCZ$ are LPLOs for stacks of three 3D surface codes. We also show that the mapping between 3D surface codes and 3D color codes described in~\cite{Kubica2015Unfolding} can be achieved using code concatenation. This result generalizes the code concatenation transformations for 2D surface codes and 2D color codes presented in~\cite{Criger2016Noise}. The second focus of this article is on quantum computing architectures based on 3D surface codes. We propose a hybrid 2D-3D surface code architecture and a purely 3D surface code architecture. Both of these architectures use the techniques of lattice surgery~\cite{Horsman2012Surface}, which we generalize to 3D surface codes. Our architectures achieve universal quantum computation without needing magic state distillation. It is possible that these architectures may require fewer resources than 2D surface code architectures in certain systems. For example, one could imagine taking advantage of the 
connections between qubits allowed in a modular architecture~\cite{Barrett2005Efficient,Fujii2012Distributed,Nickerson2013Topological,Monroe2014Large,Nickerson2014Freely} to build a code which is local in three spatial dimensions. However, more research is necessary before we can definitively assess the resource costs of our proposed architectures.

The remainder of this article is structured as follows. We provide background information on topological codes in Section~\ref{sec:bground} and we introduce a rectified picture of 3D surface codes in Section~\ref{sec:rectified}. In Section~\ref{sec:concat}, we detail a concatenation transformation that maps three 3D surface codes to a 3D color code. In Section~\ref{sec:gates}, we show that $CZ$ and $CCZ$ are transversal for stacks of three 3D surface codes and we explain how to implement a universal gate set. In Sections~\ref{sec:lattice surgery} and \ref{sec:arch}, we discuss 3D surface code lattice surgery and universal quantum computing architectures which utilize 3D surface codes. Finally, in Section~\ref{sec:discuss}, we discuss the implications of our work and outline future research directions. 

\section{\label{sec:bground} Background}

Surface codes are a family of topological stabilizer codes~\cite{Kitaev2003Anyons,BravyiKitaev1998Boundary,Dennis2002Topological}. A stabilizer code is a quantum error-correcting code defined by its stabilizer group $\mathcal{S}$, an abelian subgroup of the Pauli group where $-I\notin\mathcal{S}$~\cite{Gottesman1997Stabilizer}. Every encoded state $\ket{\overline{\psi}}$ in the code is stabilized by $\mathcal{S}$, that is $\forall S\in\mathcal{S}$, $S\ket{\overline{\psi}}=\ket{\overline{\psi}}$. 
 We summarize the properties of a quantum error correcting code with the shorthand notation $[[n,k,d]]$, where $n$ is the number of physical qubits, $k$ is the number of encoded logical qubits and $d$ is the code distance. The code distance of a quantum error-correcting code is equal to the weight of the minimum weight logical operator of the code. The weight of an operator is simply the number of qubits it act on non-trivially. 

A topological code is a code defined on some lattice with physical qubits placed on some of the elements of the lattice (the edges, for example). The stabilizers of a topological code act in geometrically local regions and the logical operators of the code form topologically non-trivial paths or surfaces on the lattice. We are particularly interested in two types of logical operators in topological codes: locality-preserving logical operators (LPLOs) and transversal logical operators. LPLOs are operators that map errors in some region of a code $R$ to errors in a region $R'$ which is at most a constant size $C$ bigger than $R$~\cite{Webster2018Locality}. A transversal logical operator is a logical operator realized by a quantum circuit of depth one which does not couple physical qubits in the same code (block). Transversal logical operators are LPLOs because transversal logical operators never spread errors from one physical qubit to another physical qubit in the same code.


In this article, we consider 2D and 3D surface codes. We begin by defining 2D surface codes in what we call the `Kitaev picture'. This is the formalism introduced by Kitaev in~\cite{Kitaev2003Anyons}. We place qubits on the edges of a 2D lattice. We associate $Z$ stabilizers with the faces of the lattice and $X$ stabilizers with the vertices of the lattice. That is, for each face $f$ we have a stabilizer $S_{f}=\bigotimes_{e\in f}Z(q_{e})$ where $Z(q_{e})$ denotes a $Z$ operator applied to the qubit on edge $e$. Analogously, for each vertex $v$ we have a stabilizer $S_{v}=\bigotimes_{e:v\in e}X(q_{e})$. We interpret unsatisfied stabilizers (stabilizers with $-1$ eigenvalues) as quasiparticles. Following the convention in the literature~\cite{Dennis2002Topological}, we refer to unsatisfied $X$ stabilizers as electric charges ($e$) and unsatisfied $Z$ stabilizers as magnetic fluxes ($m$). In the 2D surface code, both $e$ and $m$ quasiparticles are zero-dimensional (0D) objects. In the bulk of the lattice, we can only create or destroy pairs of quasiparticles of the same type. 2D surface codes can have two types of boundary: rough boundaries and smooth boundaries. Single $e$ quasiparticles can condense on the rough boundaries and single $m$ quasiparticles can condense on the smooth boundaries. In this context, quasiparticle condensation means that a single quasiparticle can be created or destroyed at the relevant boundary. In the 2D surface code, logical $\overline{Z}$ operators are strings of $Z$ operators from one rough boundary another and logical $\overline{X}$ operators are strings of $X$ operators from one smooth boundary to another. 
\begin{figure}[ht]
    \centering
    \includegraphics[width=0.7\columnwidth]{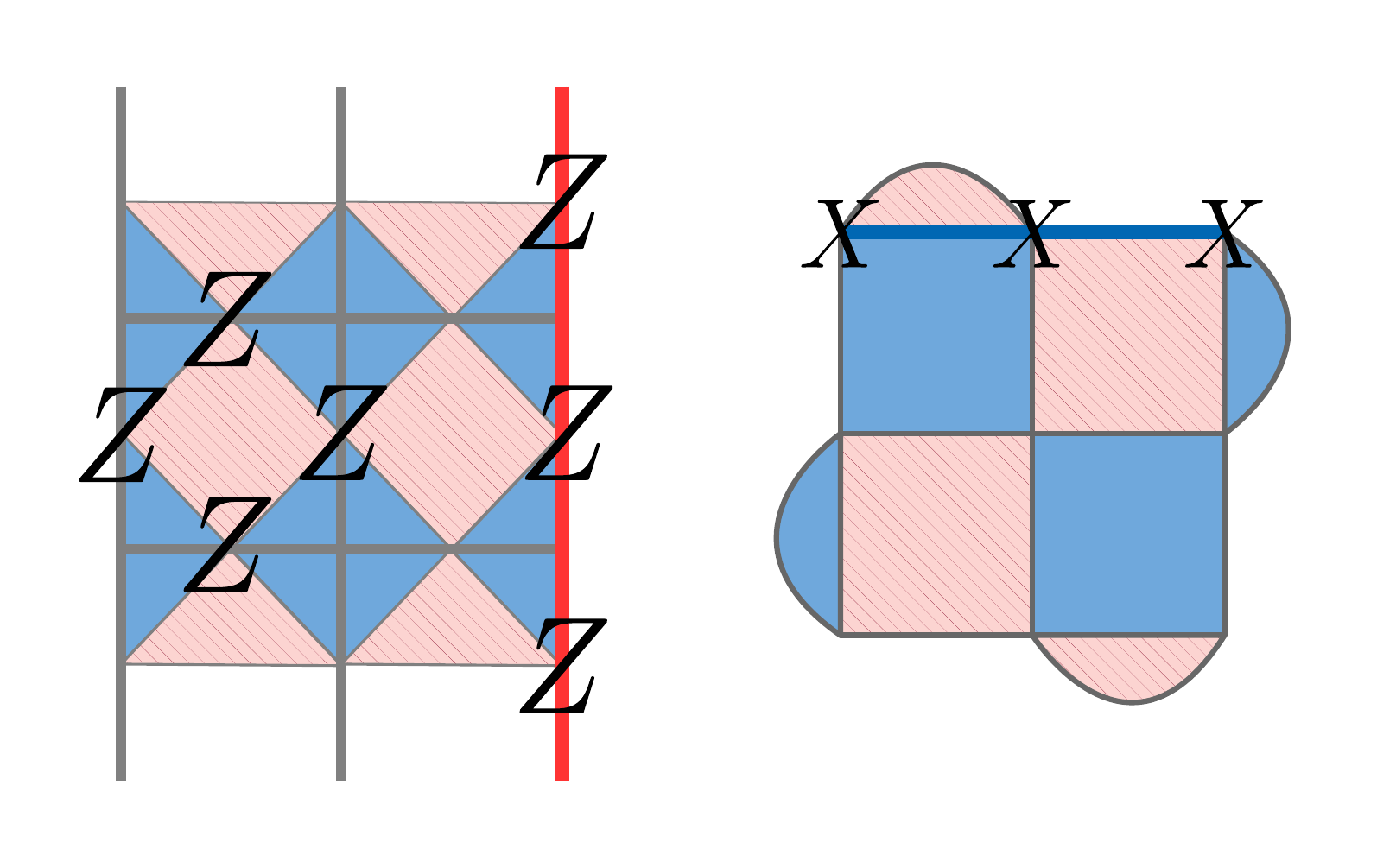}
    \caption{\label{fig:2dsc} 
    2D surface codes in the Kitaev picture and the rotated picture. On the left we show the [[13,1,3]] surface code in the Kitaev picture and the rotated picture (red (hatched) and blue (solid) lattice). We highlight a $Z$ stabilizer and $\overline{Z}$ operator. On the right we show the [[9,1,3]] surface code in the rotated picture. We highlight a $\overline{X}$ operator. In both codes, the top/bottom boundaries are rough boundaries and the left/right boundaries are smooth boundaries. 
    }
\end{figure}

There is an equivalent picture of 2D surface codes which is related to the Kitaev picture by a medial transformation~\cite{Wen2003QuantumOrders,Bombin2007Rotated,Anderson2013Homological}. This picture is often called the rotated picture~\cite{Tomita2014Low}. In the rotated picture, qubits are on vertices and stabilizers are associated with faces. Rotated picture lattices are 2-face-colourable \emph{i.e}.\ every face in the lattice can be assigned one of two colours such that no faces which share an edge have the same colour. In this picture, we associate $Z$ stabilizers with $c$ coloured faces ($c$-faces) and $X$ stabilizers with $c'$-faces. For example, the stabilizer associated with the $c$-face $f_{c}$ is $S_{f_{c}}=\bigotimes_{v\in f_{c}}Z(v)$. Figure~\ref{fig:2dsc} shows two distance three 2D surface codes in the Kitaev picture and the rotated picture.


We now turn to 3D surface codes. Initially, we define 3D surface codes in the Kitaev picture~\cite{Dennis2002Topological}, using the same conventions as the 2D surface code. We place qubits on the edges of a 3D lattice, we associate $X$ stabilizers with the vertices of the lattice and we associate $Z$ stabilizers with the faces of the lattice. We again interpret unsatisfied $X$ ($Z$) stabilizers as $e$ ($m$) quasiparticles. However, in contrast to 2D surface codes, in the 3D surface code $m$ quasiparticles are 1-D objects ($e$ quasiparticles are still 0-D). 3D surface codes also have rough and smooth boundaries which are again defined by quasiparticle condensation. As in the 2D case, $e$ ($m$) quasiparticles can condense on rough (smooth) boundaries. $\overline{Z}$ operators in 3D surface codes are strings of $Z$ operators which terminate at different rough boundaries. $\overline{X}$ operators are membranes of $X$ operators with a boundary which spans contiguous smooth boundaries. In this article we only consider 3D surface codes with six boundaries (two rough boundaries and four smooth boundaries) where the rough boundaries are on opposite sides of the lattice. 

So far we have only discussed the primal lattice picture of 3D surface codes. We can also analyse 3D surface codes in the dual lattice picture. Given a 3D lattice, we can construct its dual using a simple procedure. First, we create vertices at the centre of the cells of the original lattice. Next, we join these new vertices with edges if their corresponding cells in the original lattice shared a face. Finally, we delete the original (primal) lattice. This transformation maps vertices to cells, edges to faces, faces to edges and cells to vertices. Therefore, in the dual lattice picture of 3D surface codes, qubits are placed on the faces, $X$ stabilizers are associated with cells and $Z$ stabilizers are associated with edges. In the remainder of this article, we will use both the primal lattice picture and the dual lattice picture to analyse 3D surface codes. 

\section{\label{sec:rectified}Rectified Picture}

In this section, we describe a picture which we use to analyse stacks of 3D surface codes. We call this picture the `rectified picture'. This picture is a generalisation of the rotated picture of 2D surface codes and it is similar to the primal lattice picture of 3D color codes~\cite{Bombin2007Topological}. We start with a 3D surface code primal lattice in the Kitaev picture. To transform to the rectified picture, we rectify the primal lattice. A rectification (or full truncation) is a geometric transformation where the edges of a lattice are truncated to points~\cite{Coxeter1973Polytopes}. Specifically, to perform a rectification, we use the following procedure. First we create new vertices at the midpoints of the edges of the original lattice. Next, we join these new vertices with edges if their corresponding edges in the original lattice were part of the same face. Finally, we delete the original lattice to obtain the rectified lattice. Under a rectification, edges are mapped to vertices, cells and vertices are mapped to cells, and faces are mapped to faces. Therefore, in the rectified picture, qubits are on vertices, $X$ stabilizers are associated with cells and $Z$ stabilizers are associated with faces. We note that there is an analogous transformation which maps a 3D surface code Kitaev picture dual lattice to a rectified picture lattice. This transformation is called a face-rectification and is equivalent to taking the dual of every cell in the lattice. Given a polyhedral cell, we construct its dual by creating vertices at the centre of the original polyhedron's faces. We then connect these vertices with edges if their corresponding faces in the original polyhedron share an edge. Figure~\ref{fig:duals} shows a cube and a cuboctahedron along with their dual polyhedra. 

\begin{figure}
    \begin{minipage}{0.35\columnwidth}
        \topinset{\textbf{a)}}{\includegraphics[width=\linewidth]{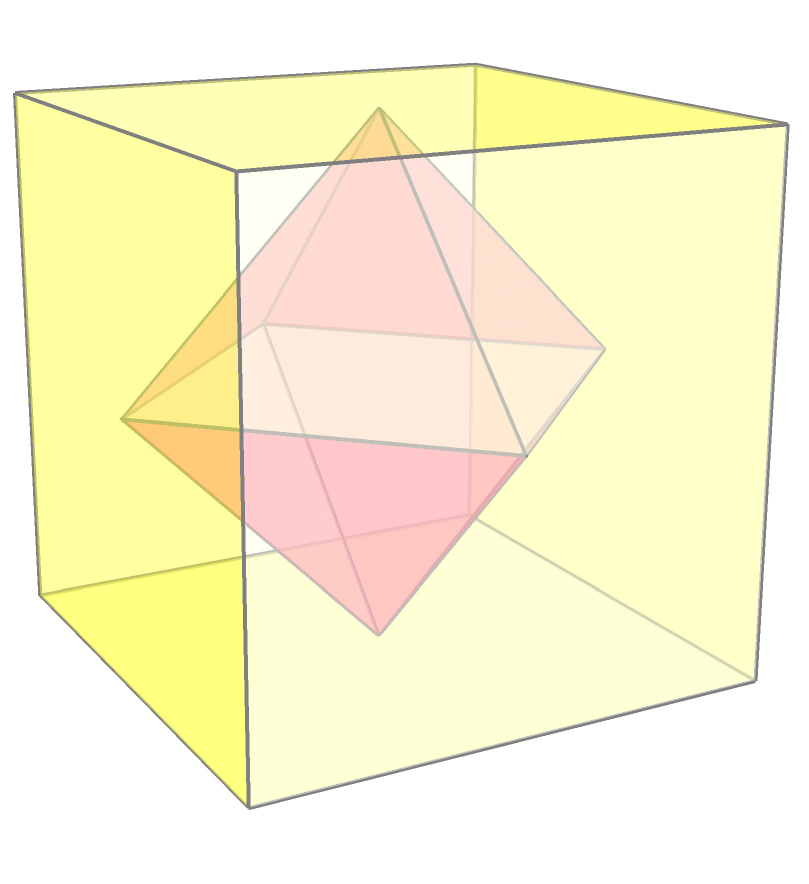}}{-0.25cm}{0cm}
    \end{minipage}
    \begin{minipage}{0.55\columnwidth}
        \topinset{\textbf{b)}}{\includegraphics[width=\linewidth]{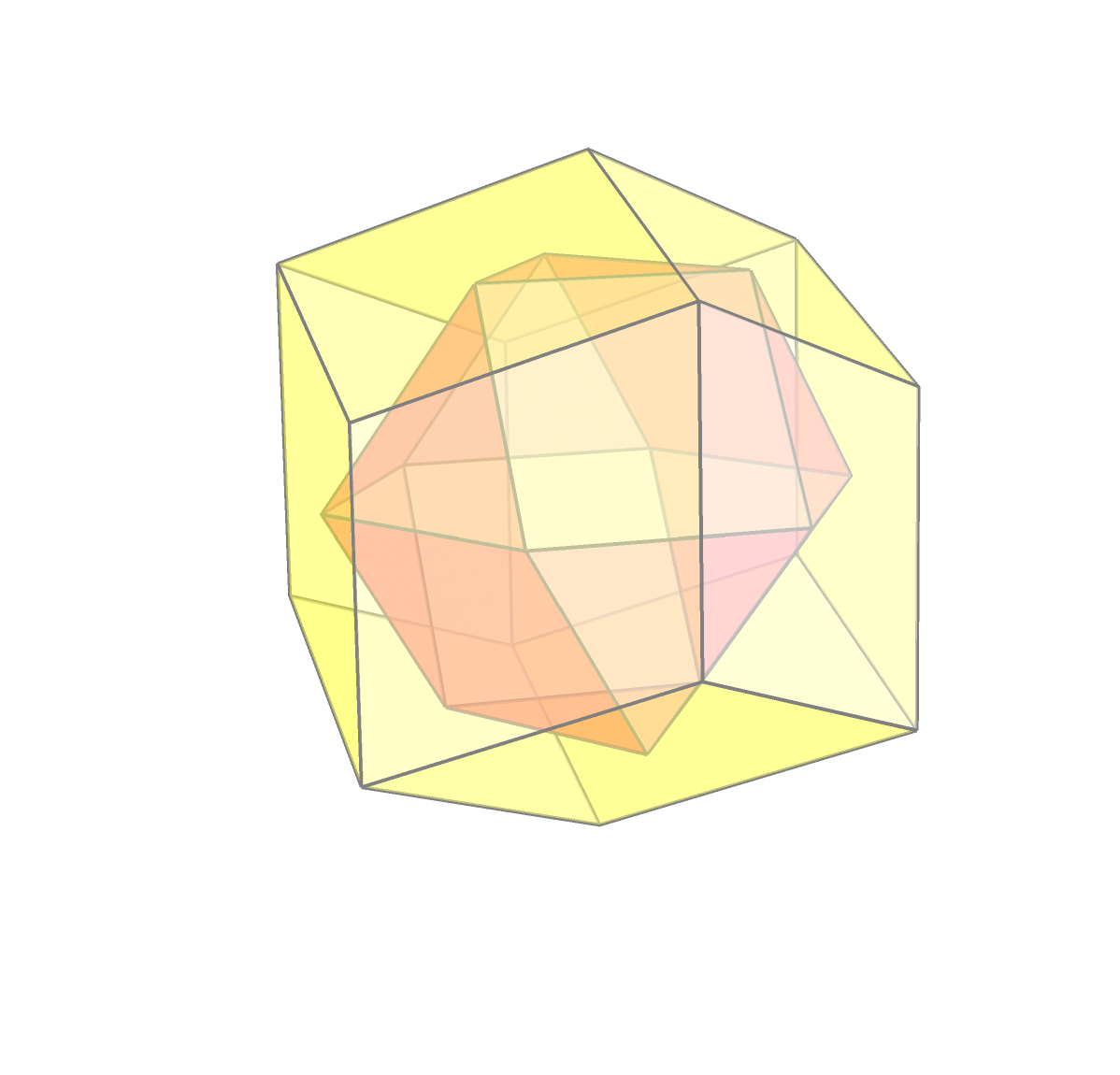}}{0.1cm}{0cm}
    \end{minipage}
    \caption{\label{fig:duals} Polyhedra and their duals. 
    \textbf{a)} A cube (yellow (light grey)) and its dual octahedron (red (medium grey)).
    \textbf{b)} A rhombic dodecahedron (yellow (light grey)) and its dual cuboctahedron (red (medium grey)).}
\end{figure}

The utility of the rectified picture comes when we consider stacks of three 3D surface codes. This is because different lattices in the Kitaev picture correspond to the same lattice in the rectified picture. Hence, instead of analysing three overlapping surface code lattices in the Kitaev picture, we can analyse a single lattice in the rectified picture. In this article, we concentrate on surface codes defined on cubic lattices and tetrahedral-octahedral lattices (primal lattices in the Kitaev picture). In the familiar cubic lattice, eight cubes meet at every vertex. In the tetrahedral-octahedral lattice, six octahedra and eight tetrahedra meet at every vertex. The cubic lattice is self-dual and the dual of a tetrahedral-octahedral lattice is a rhombic dodecahedral lattice (a lattice where every cell is a rhombic dodecahedron). Let us consider how the cubic lattice is transformed under rectification: vertices are mapped to octahedra and cubes are mapped to cuboctahedra. Cuboctahedra are polyhedra with 12 vertices where two triangle faces and two square faces meet at each vertex. Figure~\ref{fig:duals}b shows a cuboctahedron. The rectification of a cubic lattice is usually called the rectified cubic lattice. In a rectified cubic lattice two octahedra and four cuboctahedra meet at every vertex. Figure~\ref{fig:rectified cubic} shows a portion of a rectified cubic lattice. 

Next, we consider the rectification of a tetrahedral-octahedral lattice. Under rectification, octahedra transform into cuboctahedra, tetrahedra transform into octahedra and vertices become cuboctahedra. Therefore, rectification also transforms the tetrahedral-octahedral lattice into a rectified cubic lattice. In fact, we can arrange one cubic lattice and two tetrahedral-octahedral lattices in such a way that all three lattices are transformed into the exact same rectified cubic lattice under rectification. To see how this works, it is easiest to consider the dual lattices and the face-rectification transformation. As we mentioned earlier, the dual of a tetrahedral-octahedral lattice is a rhombic dodecahedral lattice. A rhombic dodecahedron is a polyhedron with twelve rhombic faces. Rhombic dodecahedra have two different types of vertex. Acute vertices are the points where the acute angle corners of four rhombi meet whereas obtuse vertices are the points where the obtuse angle corners of three rhombi meet. Figure~\ref{fig:duals}b shows a rhombic dodecahedron. In a rhombic dodecahedral lattice, four rhombic dodecahedra meet at every obtuse vertex and six rhombic dodecahedra meet at every acute vertex.

\begin{figure}
    \includegraphics[width=0.7\columnwidth]{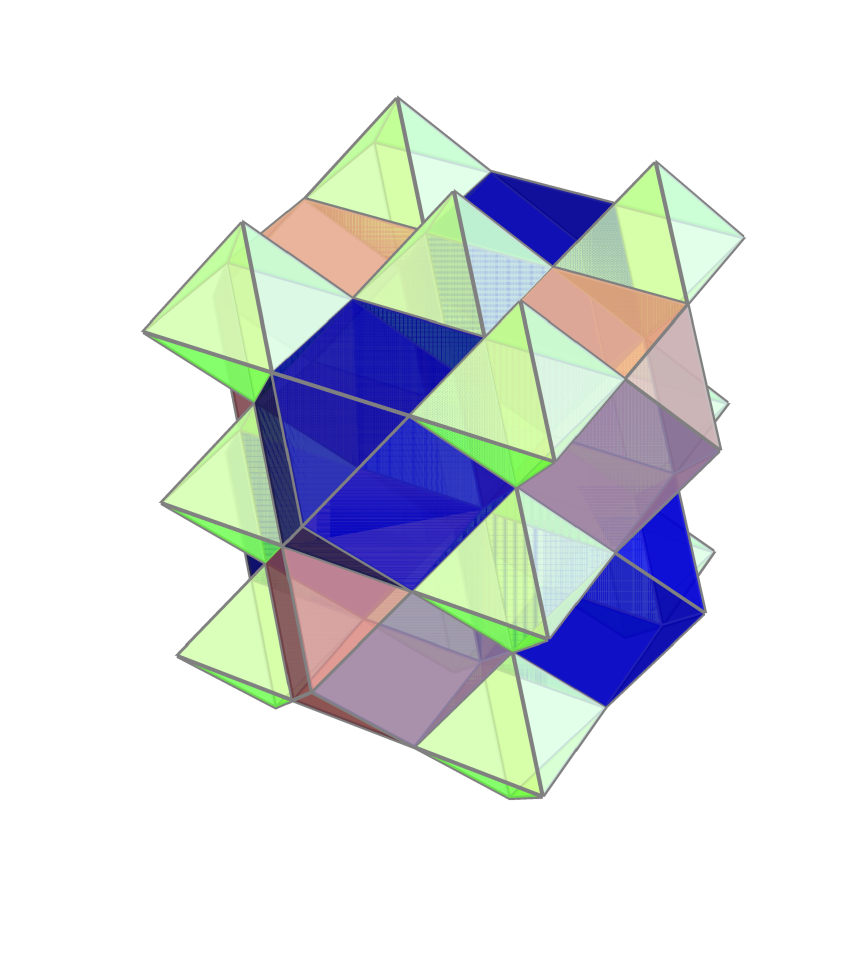}
    \caption{\label{fig:rectified cubic} Part of a rectified cubic lattice. Rectified cubic lattices consist of cuboctahedra (blue (dark grey) and red (medium grey)) and octahedra (green (light grey)). Four cuboctahedra and two octahedra meet at each vertex.}
\end{figure}

We now show how to arrange one cubic lattice and two rhombic dodecahedral lattices such that they are mapped to exactly the same lattice under face-rectification. This fact is the reason we can define three surface codes on the same rectified cubic lattice. We assume the three lattices are infinite for simplicity. First we note that the cubic lattice is 2-vertex-colourable \emph{i.e}.\ all the vertices in the cubic lattice can be assigned a colour such that no vertices which share an edge have the same colour. We give these two sets of vertices the labels $a$ and $b$. We arrange the cubic lattice and one of the rhombic dodecahedral lattices such that the acute vertices of the rhombic dodecahedra occupy the same positions as the $a$ vertices of the cubes. In this arrangement the obtuse vertices of the rhombic dodecahedra are at the centre of cubes and the $b$ vertices of the cubes are at the centre of rhombic dodecahedra. This layout is shown (for a single cube and rhombic dodecahedron) in Figure~\ref{fig:lattice arrangement}a. Next we add a second rhombic dodecahedral lattice and arrange it such that its acute vertices occupy the same positions as the $b$ vertices of the cubes. The arrangement of all three lattices is illustrated in Figure~\ref{fig:lattice arrangement}b.

\begin{figure}
    \centering
    \begin{minipage}{0.48\columnwidth}
        \centering
        \topinset{\textbf{a)}}{\includegraphics[width=\linewidth]{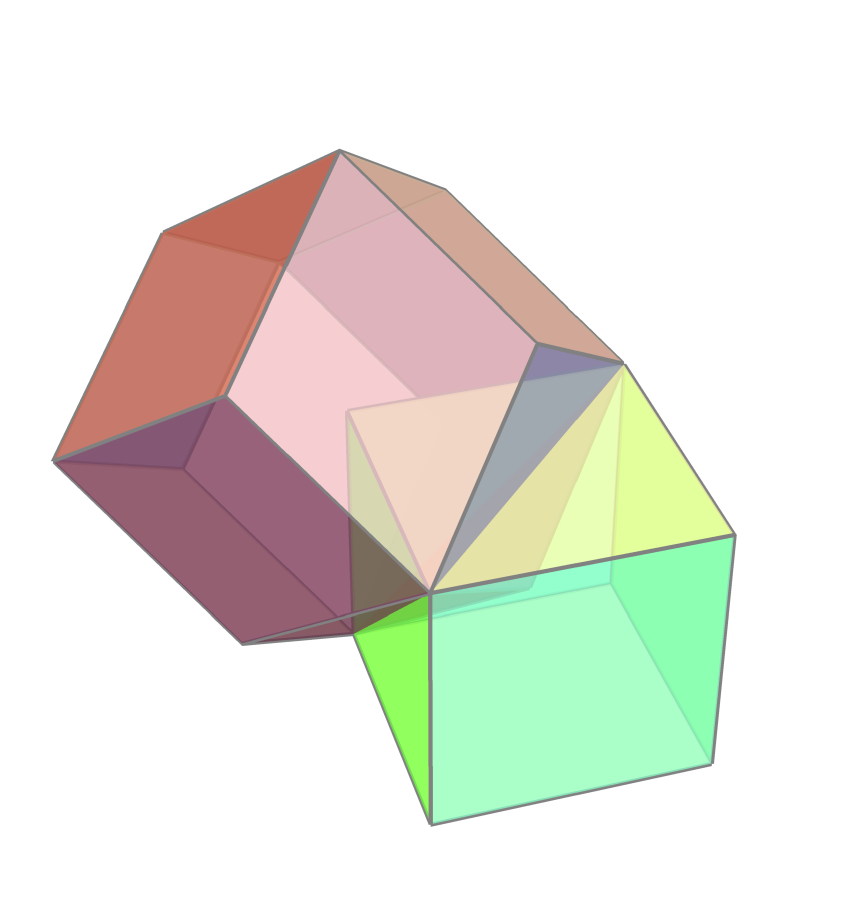}}{-0.3cm}{0cm}
    \end{minipage}
    \begin{minipage}{0.5\columnwidth}
        \centering
        \topinset{\textbf{b)}}{\includegraphics[width=\linewidth]{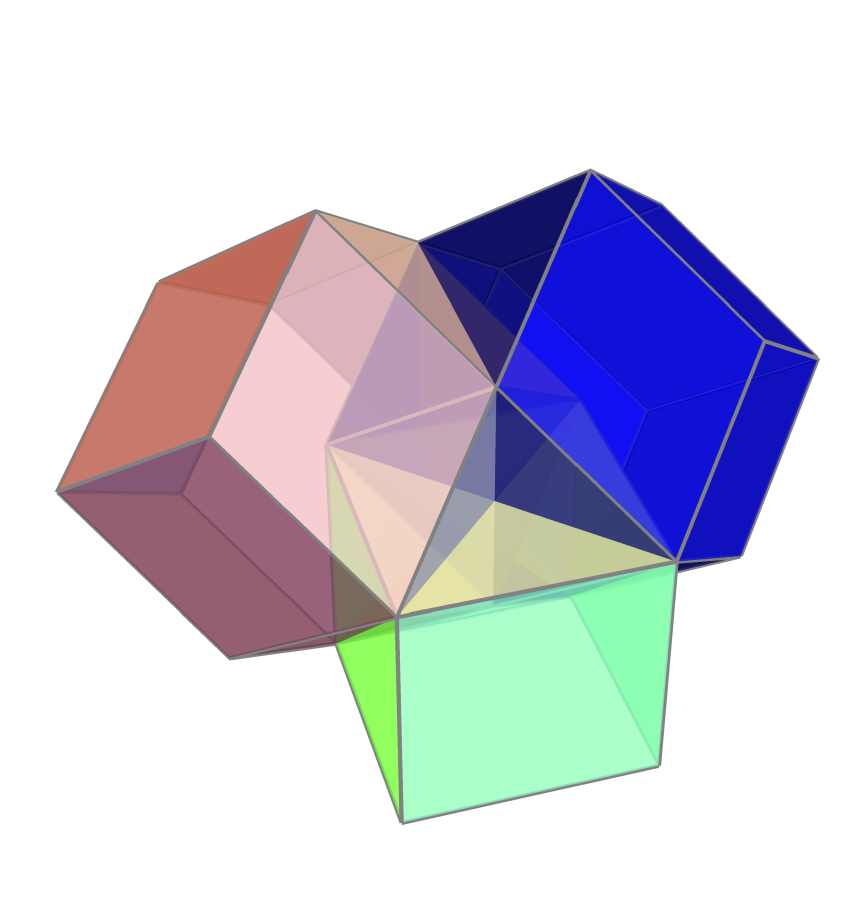}}{0.05cm}{0cm}
    \end{minipage}
    \caption{\label{fig:lattice arrangement} Arranging a cubic lattice and two rhombic dodecahedral lattices such that they are transformed to the same rectified cubic lattice under face-rectification. We show a single cell from each lattice.
    \textbf{a)} We arrange a rhombic dodecahedral lattice (red (medium grey)) and a cubic lattice (green (light grey)) such that the acute vertices of the rhombic dodecahedra occupy the same locations as the $a$-vertices of the cubes. In this arrangement the $b$-vertices of the cubes lie at the centre of the rhombic dodecahedra and the obtuse vertices of the rhombic dodecahedra lie at the centre of the cubes.
    \textbf{b)} We add a second rhombic dodecahedral lattice (blue (dark grey)) and arrange it such that the acute vertices of the rhombic dodecahedra occupy the same locations as the $b$-vertices of the cubes. In this arrangement the obtuse vertices of the rhombic dodecahedra lie at the centre of the cubes and the $a$-vertices of the cubes lie at the centre of the rhombic dodecahedra.
    }
\end{figure}

In the arrangement of lattices we have just described, all three lattices will be mapped to an identical rectified cubic lattice by the face-rectification transformation. To see why this is true we consider how the cells and vertices of the lattices transform. The cubes and the obtuse vertices of the rhombic dodecahedra both transform into octahedra. The obtuse vertices of the rhombic dodecahedra lie at the centre of the cubes in our arrangement so each lattice transforms in the same way at these positions. Similarly, the acute vertices of one rhombic dodecahedral lattice occupy the same position as the $a$-vertices of the cubic lattice. Both these types of vertices lie at the centre of the cells of the other rhombic dodecahedral lattice. Rhombic dodecahedra, acute vertices and vertices of cubes are all mapped to cuboctahedra under face-rectification. So the three lattices transform in the same way at these positions. An identical argument holds for the $b$-vertices of the cubes. 

\subsection{\label{subsec:code family} A family of stacked 3D surface codes}

In this section, we define a family of stacked 3D surface codes. We call these codes rectified cubic codes. Each member of the family consists of three 3D surface codes supported on the same rectified cubic lattice. We first discuss the structure of the rectified cubic lattices then we define the surface codes.

\subsubsection{\label{subsubsec:lattice structure} Lattice structure}

Rectified cubic lattices are three-cell-colourable, and we colour the cells of our lattices with the colours $\{r,g,b\}$. We assume that octahedra are coloured $g$ and the two sets of cuboctahedra are coloured $r$ and $b$. We assign each lattice face the colour of the two cells it is part of. For example, a face shared by a $r$-cell and a $g$-cell is a $rg$-face. A face on an a boundary which is only part of one cell is assigned the combination of colours it would have in a infinite lattice. The lattices in our family have two types of boundary. One type of boundary slices a layer of cuboctahedra in half and the other type of boundary slices between a layer of cuboctahedra. We call these boundaries half cuboctahedra boundaries and full cuboctahedra boundaries, respectively. Each lattice in the family has two half cuboctahedra boundaries and four full cuboctahedra boundaries. Opposite boundaries are the same type. The two types of boundary are shown in Figure~\ref{fig:d3 lattice}.

\begin{figure}
    \includegraphics[width=0.6\columnwidth]{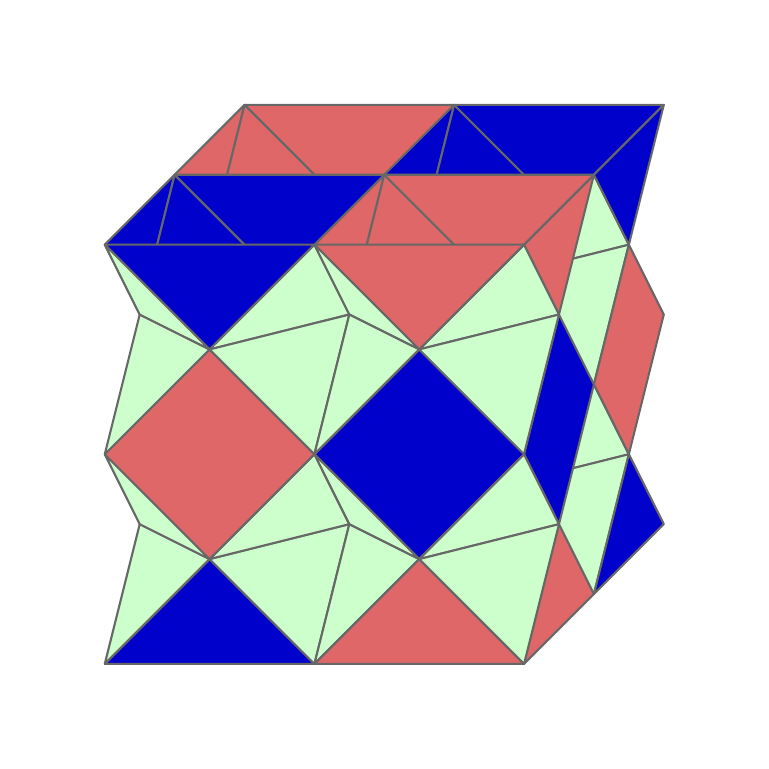}
    \caption{\label{fig:d3 lattice} The $d=3$ rectified cubic lattice. The top and bottom boundaries are half cuboctahedra boundaries whereas the other four boundaries are full cuboctahedra boundaries.}
\end{figure} 

We parameterize the lattices in our family by a parameter $d$ which will be equal to the code distance of the three codes supported on a particular lattice. We specify the structure of a distance $d$ lattice by dividing it into 2D layers which are parallel to the half cuboctahedra boundaries. There are two types of layer in this division, which we call `chequerboard layers' and `diamond layers', due to their appearance. Figure~\ref{fig:layers} shows the structure of the two types of layer in the $d=3$ lattice and the $d=4$ lattice. In a distance $d$ lattice, there are $d$ chequerboard layers and $d-1$ diamond layers and the two types of layer alternate. The half cuboctahedra boundaries are themselves chequerboard layers. Layers directly above and below each other are connected by edges as can be seen in Figure~\ref{fig:d3 lattice}.  

\begin{figure}
    \includegraphics[width=0.8\columnwidth]{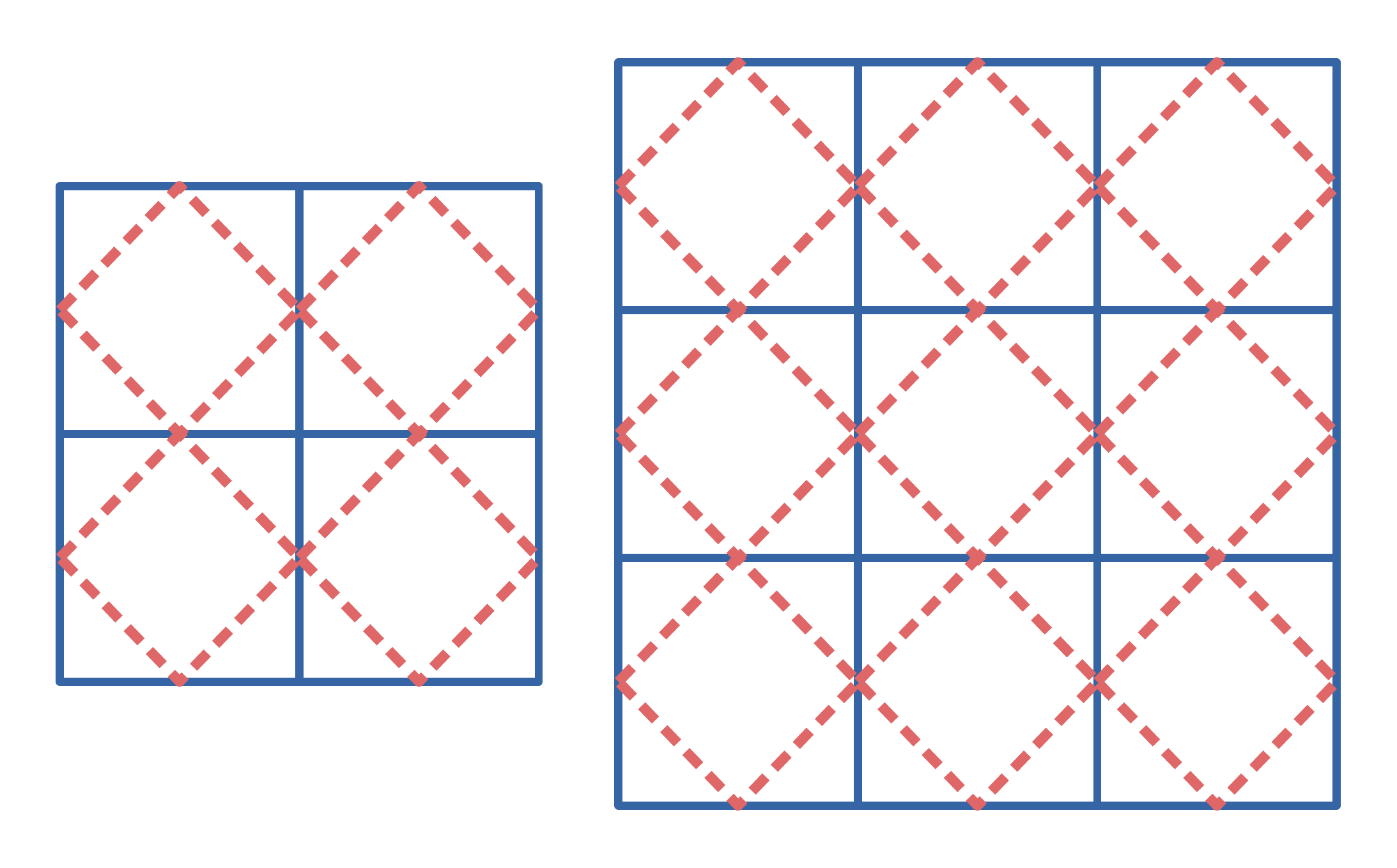}
    \caption{\label{fig:layers} The two types of layer in a $d=3$ rectified cubic lattice (left) and a $d=4$ rectified cubic lattice (right). Chequerboard layers (continuous blue lines) are layers which slice cuboctahedra in half and diamond layers (dashed red lines) are layers which slice octahedra in half.}
\end{figure}

\subsubsection{\label{subsubsec:code structure} Code structure}

In this section, we specify the structure of the three 3D surface codes defined on the same distance $d$ rectified cubic lattice. We place three qubits at each vertex of the lattice (one qubit per code). Each chequerboard layer in the lattice has $d^{2}$ vertices each and each diamond layer has $2d(d-1)$ vertices. Therefore, for a distance $d$ lattice, the number of physical qubits in each code is:
\begin{equation}
    \begin{split}
        n&=d^{3}+2d(d-1)^{2}, \\
        &=3d^{3}-4d^{2}+2d.
    \end{split}
    \label{eq:q count}
\end{equation}

We label each code, $\mathcal{SC}_{c}$, with the colour of its $X$ stabilizers. $\mathcal{SC}_{c}$ has $X$ stabilizers associated with $c$-cells and $Z$ stabilizers associated with $c'c''$-faces. The following Table and Figure detail the stabilizers of the three codes supported on a rectified cubic lattice:

\begin{center}
    \begin{tabular}{||c | c | c||} 
    \hline
    Code & X stabilizers & Z stabilizers \\ [0.5ex] 
    \hline
    $\mathcal{SC}_{r}$ & $r$-cuboctahedra & $bg$-faces \\ 
    \hline
    $\mathcal{SC}_{g}$ & $g$-octahedra & $rb$-faces  \\
    \hline
    $\mathcal{SC}_{b}$ & $b$-cuboctahedra & $rg$-faces  \\ [1ex]
    \hline
   \end{tabular}
\end{center}

\begin{figure}[h!]
    \centering
    \begin{minipage}{0.48\columnwidth}
        \centering
        \topinset{$\mathcal{SC}_{r}$ $X$ stabilizer}{\includegraphics[width=0.8\linewidth]{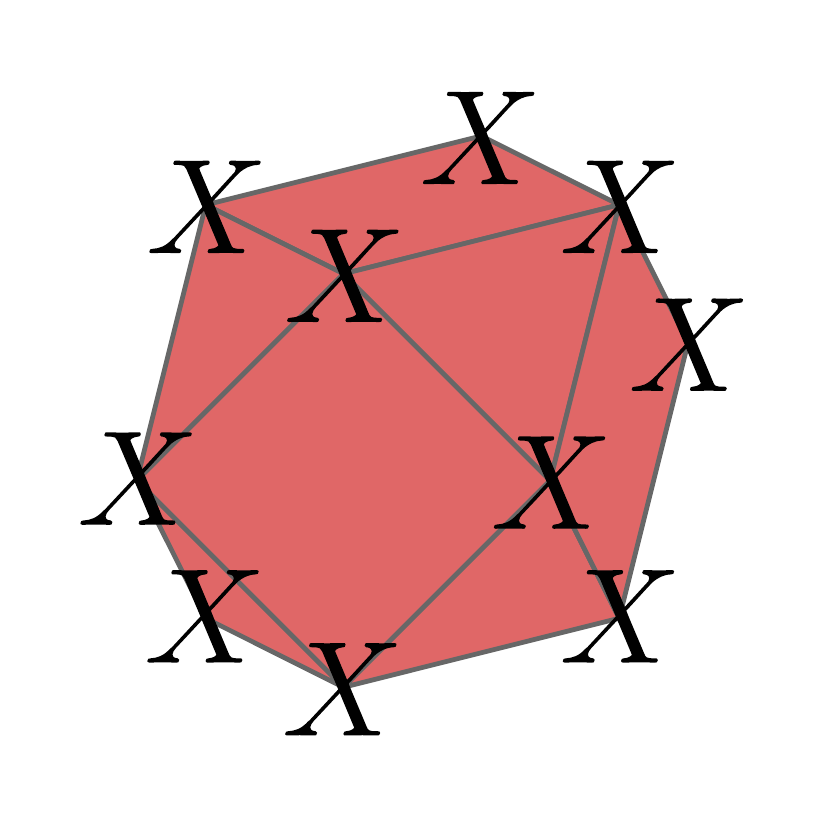}}{-0.3cm}{0cm}
    \end{minipage}
    \begin{minipage}{0.5\columnwidth}
        \centering
        \topinset{$\mathcal{SC}_{r}$ $Z$ stabilizer}{\includegraphics[width=0.8\linewidth]{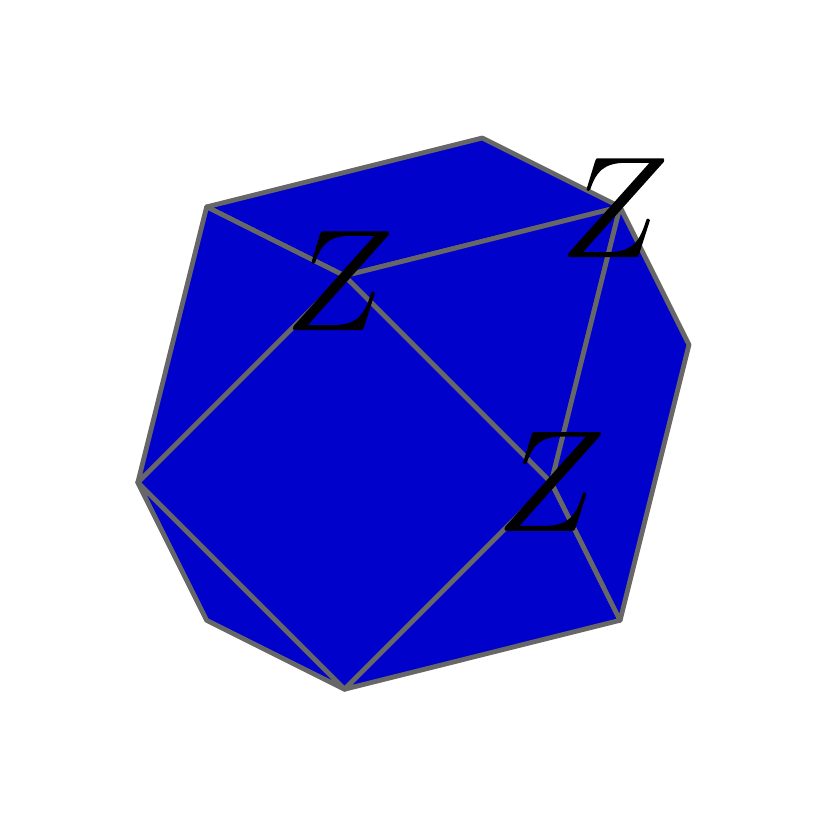}}{0.05cm}{0cm}
    \end{minipage}
    \begin{minipage}{0.48\columnwidth}
        \centering
        \topinset{$\mathcal{SC}_{g}$ $X$ stabilizer}{\includegraphics[width=0.75\linewidth]{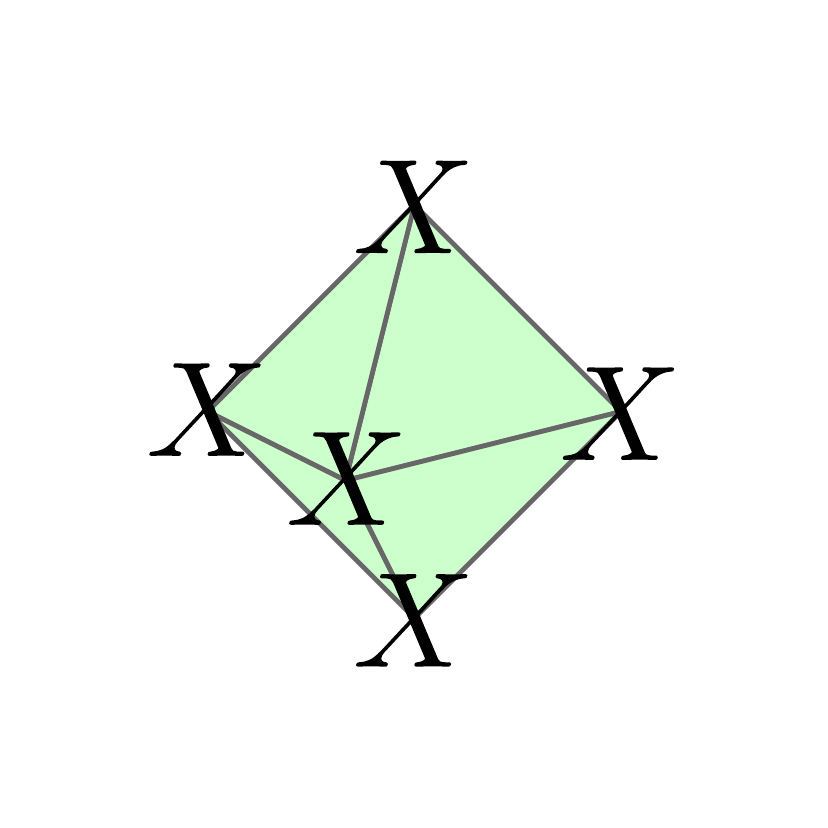}}{-0.3cm}{0cm}
    \end{minipage}
    \begin{minipage}{0.5\columnwidth}
        \centering
        \topinset{$\mathcal{SC}_{g}$ $Z$ stabilizer}{\includegraphics[width=0.8\linewidth]{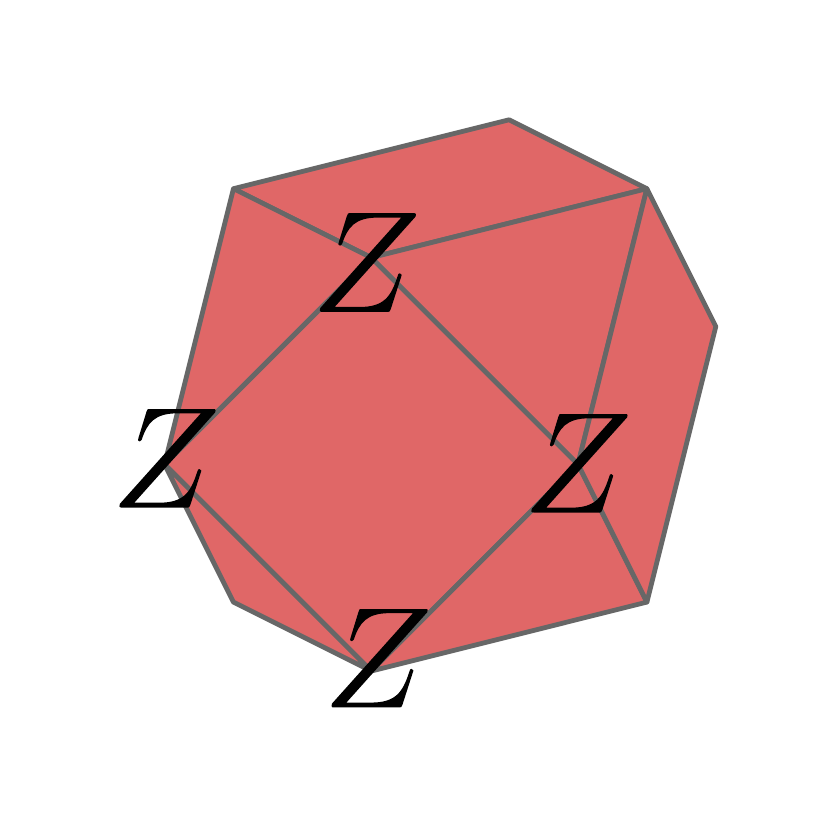}}{0.05cm}{0cm}
    \end{minipage}
    \caption{\label{fig:stabilizers} The stabilizers of $\mathcal{SC}_{r}$ and $\mathcal{SC}_{g}$. The $\mathcal{SC}_{b}$ stabilizers are identical to those of $\mathcal{SC}_{r}$ except with diagonally red (medium grey) and blue (dark grey) interchanged.
    }
\end{figure}

We also associate colours with the boundaries of our rectified cubic lattices. A $c$-boundary corresponds to a rough boundary in $\mathcal{SC}_{c}$ and smooth boundaries in $\mathcal{SC}_{c'}$ and $\mathcal{SC}_{c''}$. In Section~\ref{sec:bground}, we defined rough and smooth boundaries in terms of quasiparticle condensation. For regular lattices like the ones we consider, we can be more specific about the structure of the boundaries. In a 3D surface code defined on a cubic lattice in the Kitaev picture, each qubit in the bulk is a member of two $X$ stabilizers and four $Z$ stabilizers. Similarly, in a 3D surface code defined on a tetrahedral-octahedral lattice in the Kitaev picture, each qubit in the bulk is a member of two $X$ stabilizers and four $Z$ stabilizers. In each of these lattices, the qubits on the rough boundaries are members of a single $X$ stabilizer and the qubits on the smooth boundaries are members of between one and three $Z$ stabilizers (\emph{i.e}.\ fewer than four). We note that the parts of lattices at which two boundaries meet are part of both boundaries. 

For our family of stacked 3D surface codes to have a transversal $CCZ$ gate (see Section~\ref{sec:gates}), we need to have two boundaries of each colour and we need opposite boundaries to have the same colour. The half cuboctahedra boundaries of the distance $d$ rectified cubic lattices we detailed in the previous section are valid $g$-boundaries. However, the full cuboctahedra boundaries are neither $r$-boundaries or $b$-boundaries. The problem is that the four full cuboctahedra boundaries are identical. We need to break the symmetry between the four full cuboctahedra boundaries to turn them into valid $r$-boundaries and $b$-boundaries. We break the symmetry by adding additional low weight stabilizers to the full cuboctahedra boundaries. These stabilizers are analogous to the weight two stabilizers on the boundaries of the [[9,1,3]] 2D surface code shown in Figure~\ref{fig:2dsc}. In Figure~\ref{fig:extra stabs} we show the additional stabilizers we add to $\mathcal{SC}_{r}$ and $\mathcal{SC}_{b}$ to turn the full cuboctahedra boundaries into $r$-boundaries and $b$-boundaries. 

\begin{figure}
    \includegraphics[width=0.6\columnwidth]{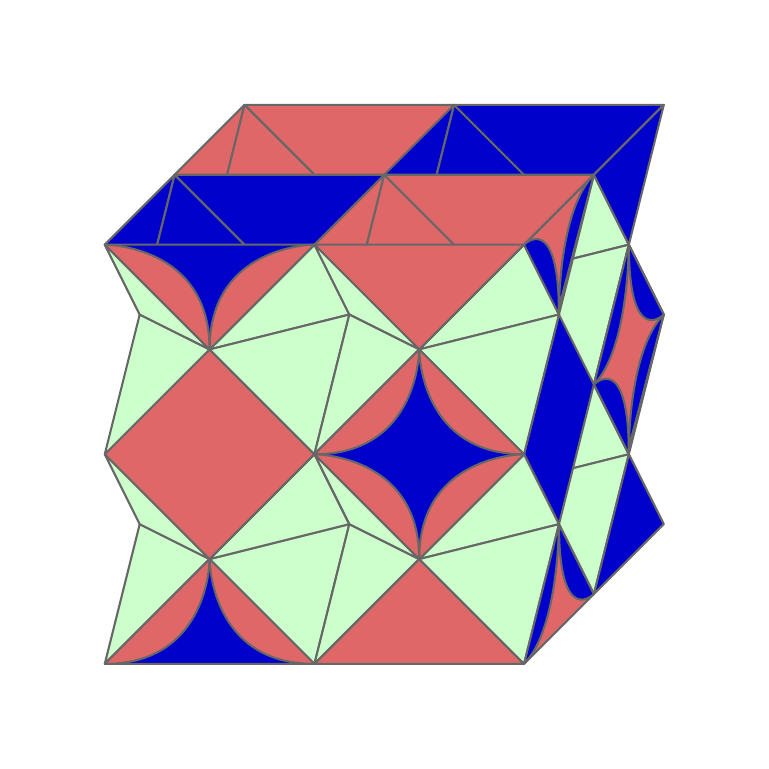}
    \caption{\label{fig:extra stabs} The additional stabilizers required such that the codes in our family of stacked 3D surface codes have the correct boundaries. First, consider the full cuboctahedra boundary facing us. We associate additional $\mathcal{SC}_{r}$ $X$ stabilizers with the faces of the $b$-cuboctahedra on this boundary (blue (dark grey) faces). In addition, we associate additional $\mathcal{SC}_{b}$ $Z$ stabilizers with some of the edges of these blue (dark grey) faces (red (medium grey) circular segments). These edges would have been part of $rg$-faces if not for the boundaries. In effect, we have added a multiple 2D flattenings of $r$-cuboctahedra to the lattice. The edges of these 2D flattenings are themselves 1D flattenings of $rg$-faces ($\mathcal{SC}_{b}$ $Z$ stabilizers). We add analogous stabilizers to the back boundary. With these additional stabilizers, the front and back boundaries are valid $b$-boundaries. Next, consider the left and right boundaries in the Figure. We associate additional $\mathcal{SC}_{b}$ $X$ stabilizers with the faces of the $r$-cuboctahedra (red (medium grey) faces) on these boundaries. We also associate additional $\mathcal{SC}_{r}$ $Z$ stabilizers with some of the edges of these faces (blue (dark grey) circular segments). With these additional stabilizers, the left and right boundaries are valid $r$-boundaries. 
    }
\end{figure}

With the additional stabilizers shown in Figure~\ref{fig:extra stabs}, we claim that the three codes have the correct structure on the boundaries of the lattice. That is, the $c$-boundaries are rough boundaries in $\mathcal{SC}_{c}$ and smooth boundaries in $\mathcal{SC}_{c'}$ and $\mathcal{SC}_{c''}$. First consider the $g$-boundaries (top and bottom boundaries in Figure~\ref{fig:extra stabs}). Each vertex on the $g$-boundaries is a member of a single $g$-octahedron ($\mathcal{SC}_{g}$ $X$ stabilizer). Each vertex is also a member of four $rb$-faces ($\mathcal{SC}_{g}$ $Z$ stabilizers), except where the $g$-boundary meets the $r$-boundaries and $b$-boundaries. The $g$-boundary is, therefore, a rough boundary in $\mathcal{SC}_{g}$. Each vertex on the $g$-boundaries is a member of two $r$-cuboctahedra (including 2D flattenings shown in Figure~\ref{fig:extra stabs}) and two $b$-cuboctahedra (including 2D flattenings), except where the $g$-boundaries meet an $r$-boundary or a $b$-boundary, respectively. The vertices on the $g$-boundaries are all members of fewer than four $rg$-faces (including 1D flattenings shown in Figure~\ref{fig:extra stabs}) and fewer than four $bg$-faces (including 1D flattenings). Therefore, the $g$-boundaries are smooth boundaries in $\mathcal{SC}_{b}$ and $\mathcal{SC}_{r}$.

The next pair of boundaries we consider are the $b$-boundaries. Due to the additional stabilizers shown in Figure~\ref{fig:extra stabs}, each vertex on the $b$-boundaries is a member of two $r$-cuboctahedra (including 2D flattenings) and two $g$-octahedra, except where the $b$-boundaries meet the $r$-boundaries and $g$-boundaries, respectively. However, each vertex on the $b$-boundaries is a member of a single $b$-cuboctahedron. Every vertex on the $b$-boundaries is a member of fewer than four $rb$-faces and fewer than four $bg$-faces (including 1D flattenings). But each vertex is a member of four $rg$-faces (including 1D flattenings) except for the vertices which are also on $r$-boundaries or $g$-boundaries. Therefore, the $b$-boundaries are rough boundaries in $\mathcal{SC}_{b}$ and smooth boundaries in $\mathcal{SC}_{r}$ and $\mathcal{SC}_{g}$, as required. The argument for $r$-boundaries is identical to the argument for $b$-boundaries, except with $r$ and $b$ exchanged. In Appendix~\ref{app:parallelepiped}, we describe an alternative family of stacked 3D surface codes which are supported on rectified cubic lattices which different boundaries to the ones we have just described. 

Next, we show that each of the three codes has one encoded logical qubit. The number of encoded qubits in a stabilizer code is equal to the number of physical qubits minus the number of stabilizer generators. So we need to count the number of stabilizer generators in each of the three codes. We begin with $\mathcal{SC}_{g}$. In this code, $X$ stabilizers are associated with $g$-cells (octahedra) and $Z$ stabilizers are associated with $rb$-faces. Consider the top $g$-boundary of a distance $d$ lattice oriented the same way as the $d=3$ lattice in Figure~\ref{fig:extra stabs}. This boundary has the structure of a chequerboard layer and each vertex on this boundary is a member of a single (complete or incomplete) octahedron. Chequerboard layers have $d^{2}$ vertices so we have $d^{2}$ octahedra which are situated directly below the top boundary. Every other chequerboard layer (except the bottom layer) also has $d^{2}$ octahedra situated below it. There are $d$ chequerboard layers so there are $d^{2}(d-1)$ octahedra in a distance $d$ lattice. The $X$ stabilizers we associate with these octahedra are all independent. Therefore, the number of $X$ stabilizer generators in $\mathcal{SC}_{g}$ is:
\begin{equation}
    \rank (S_{X}^{(g)})=d^{2}(d-1).
    \label{eq:scg xs}
\end{equation}

We now count the $Z$ stabilizer generators of $\mathcal{SC}_{g}$. As we stated previously, these stabilizers are associated with the $rb$-faces of the lattice. We split these faces into two groups: faces which are parallel to $g$-boundaries, and faces which are parallel to the $r$-boundaries or the $b$-boundaries. In a distance $d$ lattice, we have $(d-1)^{2}$ $rb$-faces parallel to the $g$-boundaries in each diamond layer. There are $d-1$ diamond layers, so there are $(d-1)^{3}$ $rb$-faces parallel to the $g$-boundaries. Each chequerboard layer cuts through $2d(d-1)$ $rb$-faces which are parallel to the $r$-boundaries or the $b$-boundaries. There are $d$ chequerboard layers, so there are $2d^{2}(d-1)$ of these $rb$-faces. Therefore, the total number of $rb$-faces in a distance $d$ lattice is $(d-1)(3d^{2}-2d+1)$. However, these stabilizers are not all independent. We can multiply the $Z$ stabilizers associated with the $rb$-faces of any cuboctahedron (both full cuboctahedra and half cuboctahedra) to get the identity. Consequently, we must remove one $Z$ stabilizer from the list of stabilizer generators for every cuboctahedron in the lattice to get a set of independent $Z$ generators. Each chequerboard layer has $(d-1)^{2}$ cuboctahedra and there are $d$ chequerboard layers, so in total we have $d(d-1)^{2}$ cuboctahedra in a distance $d$ lattice. Therefore, the total number of $Z$ stabilizer generators in $\mathcal{SC}_{g}$ is:
\begin{equation}
    \rank (S_{Z}^{(g)})=(d-1)(2d^{2}-d+1).
    \label{eq:scg zs}
\end{equation}

The total number of stabilizer generators in $\mathcal{SC}_{g}$ is therefore:
\begin{equation}
    \begin{split}
    \rank (S_{X}^{(g)})+\rank(S_{Z}^{(g)})&=(d-1)(3d^{2}-d+1), \\
    &=3d^{3}-4d^{2}+2d-1.
    \end{split}
    \label{eq:scg stot}
\end{equation}
By comparing Equations~\ref{eq:scg stot} and \ref{eq:q count}, we see that $\mathcal{SC}_{g}$ has has $n-1$ stabilizer generators, where $n$ is the number of physical qubits in the code. Therefore, $\mathcal{SC}_{g}$ encodes a single logical qubit.

Next, we count the stabilizer generators of $\mathcal{SC}_{b}$. The $X$ stabilizers of this code are associated with $b$-cells (including the 2D flattenings) and the $Z$ stabilizers are associated with $rg$-faces (including the 1D flattenings). First, we count the $X$ stabilizers of $\mathcal{SC}_{b}$. Consider the chequerboard layers parallel to the $g$-boundaries. Each chequerboard layer has $(d-1)^{2}$ cuboctahedra (half of which are $r$ and half of which are $b$). There are $d$ chequerboard layers, so there are $d(d-1)^{2}/2$ $\mathcal{SC}_{b}$ $X$ stabilizers associated with $b$-cuboctahedra (either full cuboctahedra or half cuboctahedra). Now consider the $r$-boundaries of the lattice. On each $r$-boundary we have additional $\mathcal{SC}_{b}$ $X$ stabilizers associated with the faces of $r$-cuboctahedra (as explained in Figure~\ref{fig:extra stabs}). There are $d(d-1)$ of these faces in a distance $d$ lattice so we have $d(d-1)$ additional $\mathcal{SC}_{b}$ $X$ stabilizers. The stabilizers we have just detailed are all independent. Hence, the total number of $X$ stabilizer generators in $\mathcal{SC}_{b}$ is:
\begin{equation}
    \rank(S_{X}^{(b)})=\frac{(d-1)}{2}(d^{2}+d).
    \label{eq:scb xs}
\end{equation}

Next, we count the $Z$ stabilizer generators of $\mathcal{SC}_{b}$. The $Z$ stabilizers of $\mathcal{SC}_{b}$ are associated with $rg$-faces (and their 1D flattenings). The $rg$-faces are part of $r$-cuboctahedra, which we counted in the previous paragraph. The $(d-1)^{2}/2$ half $r$-cuboctahedra on the $g$-boundaries have four $rg$-faces. The chequerboard layers which are parallel to the $g$-boundaries but are not the $g$-boundaries each have $(d-1)^{2}/2$ full $r$-cuboctahedra with eight $rg$-faces. There are $d$ chequerboard layers in a distance $d$ lattice and two of these layers are the $g$-boundaries. Therefore, the total number of $\mathcal{SC}_{b}$ $Z$ stabilizers associated with $rg$-faces is $4(d-1)^{3}$. As shown in Figure~\ref{fig:extra stabs}, we also have $\mathcal{SC}_{b}$ $Z$ stabilizers which are associated with the edges of the faces which belong to $b$-cuboctahedra on the $b$-boundaries. These faces are either square or triangular. Each square face has three independent $Z$ stabilizers associated with its edges and each triangular face has two independent $Z$ stabilizers associated with its edges. There are $2(d-1)$ triangular faces and $(d-1)(d-2)$ square faces on the $b$-boundaries which belong to $b$-cuboctahedra in a distance $d$ lattice. Therefore, the total number of independent weight two $Z$ stabilizers in $\mathcal{SC}_{b}$ is $(d-1)(3d-2)$. 

Some of the $Z$ stabilizers we have counted so far are not independent. Consider a complete octahedron. Half of its faces are $rg$-faces and half are $bg$-faces. The product of the $Z$ stabilizers associated with the $rg$-faces is the identity, as each vertex is part of exactly two $rg$-faces. The product of all the $Z$ stabilizers associated with the $rg$-faces of each complete $r$-cuboctahedron is also the identity for the same reason. Therefore we must lose a single $Z$ stabilizer from the list of stabilizer generators for each complete octahedron and $r$-cuboctahedron. Every chequerboard layer parallel to the $g$-boundaries (except the bottom $g$-boundary) has a complete octahedron below all the vertices in the bulk of the layer. There are therefore $(d-1)(d-2)^{2}$ complete octahedra in a distance $d$ lattice. We have already counted the $(d-1)^{2}(d-2)/2$ complete $r$-cuboctahedra. There is also one other redundancy we have not taken into account. We can construct the identity by multiplying the $Z$ stabilizers associated with the $rg$-faces and edges of the half octahedra on the $b$-boundaries, as illustrated in Figure~\ref{fig:halfoct redund}. There are $2(d-1)(d-2)$ of these half octahedra. In total we need to remove $(d-1)(3d^{2}-7d+2)/2$ redundant $Z$ stabilizers from the list of stabilizer generators. The total number of $Z$ stabilizer generators in $\mathcal{SC}_{b}$ is therefore:
\begin{equation}
    \begin{split}
        \rank (S_{Z}^{(b)})&=\frac{(d-1)}{2}(8(d-1)^{2}+6d-4-3d^{2}+7d-2), \\
        &=\frac{d-1}{2}(5d^{2}-3d+2).
    \end{split}
    \label{eq:scb zs}
\end{equation}

\begin{figure}
    \includegraphics[width=0.6\columnwidth]{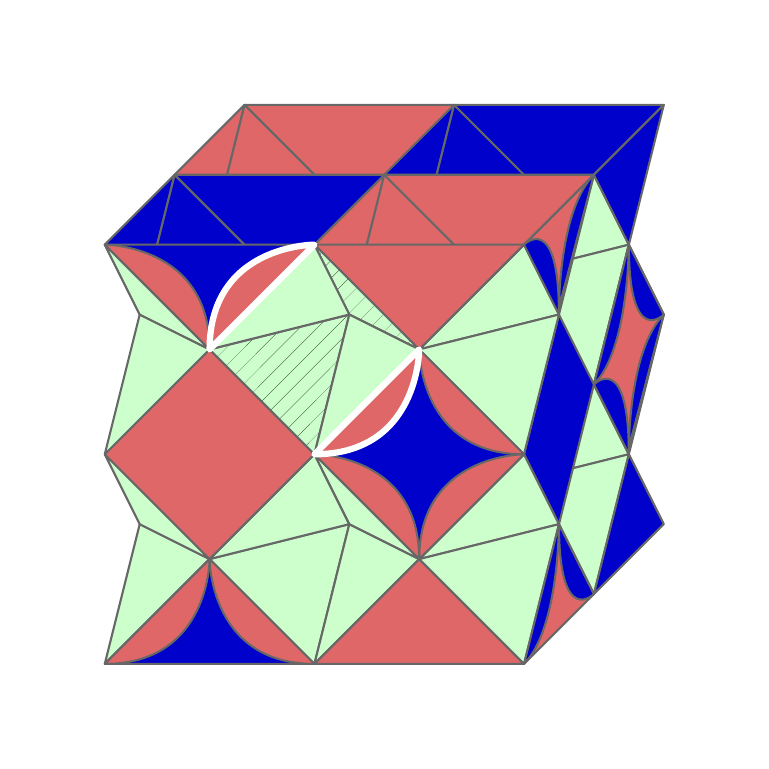}
    \caption{\label{fig:halfoct redund} Redundant $Z$ stabilizers in $\mathcal{SC}_{b}$. We can construct the identity by multiplying the $Z$ stabilizers associated with the $rg$-faces and edges of half octahedra on the $b$-boundaries. We have highlighted one such collection of faces (hatched green triangles) and circular segments (red faces with white edges).}
\end{figure}

The total number of stabilizer generators in $\mathcal{SC}_{b}$ is:
\begin{equation}
    \begin{split}
    \rank (S_{X}^{(b)})+\rank(S_{Z}^{(b)})&=(d-1)(3d^{2}-d+1), \\
    &=3d^{3}-4d^{2}+2d-1.
    \end{split}
    \label{eq:scb stot}
\end{equation}
By comparing Equations~\ref{eq:scb stot} and \ref{eq:q count}, we see that $\mathcal{SC}_{b}$ has has $n-1$ stabilizer generators, where $n$ is the number of physical qubits in the code. Therefore, $\mathcal{SC}_{b}$ encodes a single logical qubit. $\mathcal{SC}_{r}$ also encodes a single logical qubit. The argument showing this is identical to the argument for $\mathcal{SC}_{b}$, except with $r$ and $b$ swapped everywhere. In Appendix~\ref{app:d2 code}, we list the stabilizer generators of the three codes supported on a $d=2$ rectified cubic lattice. 

\subsubsection{\label{subsubsec:log ops} Logical operators}

To finish our discussion of rectified cubic codes, we detail the logical operators of the three surface codes supported on a distance $d$ rectified cubic lattice. $\overline{Z}_{c}$ operators are strings of $Z$ operators from one $c$-boundary to the other and $\overline{X}_{c}$ operators are membranes of $X$ operators with a boundary that spans the $c'$ and $c''$-boundaries. It is useful to define a canonical set of logical operators for each code. The canonical $\overline{Z}_{c}$ operators lie along the lines where $c'$-boundaries meet $c''$-boundaries. That is, given a $c'$-boundary and a $c''$-boundary that share vertices, a canonical $\overline{Z}_{c}$ operator acts on all qubits which are members of both boundaries. Figure~\ref{fig:canon log ops} shows example canonical $\overline{Z}_{c}$ operators for the three codes in a single stack. These canonical $\overline{Z}_{c}$ operators are weight $d$, where $d$ is the code distance that parameterizes the lattice. We define the canonical $\overline{X}_{c}$ operators as membranes of $X$ operators which act on every qubit on one of the $c$-boundaries. The canonical $\overline{X}_{g}$ operators are weight $d^{2}$ and the canonical $\overline{X}_{r}$ and $\overline{X}_{b}$ operators are weight $d^{2}+(d-1)^{2}$.

\begin{figure}
    \includegraphics[width=0.6\columnwidth]{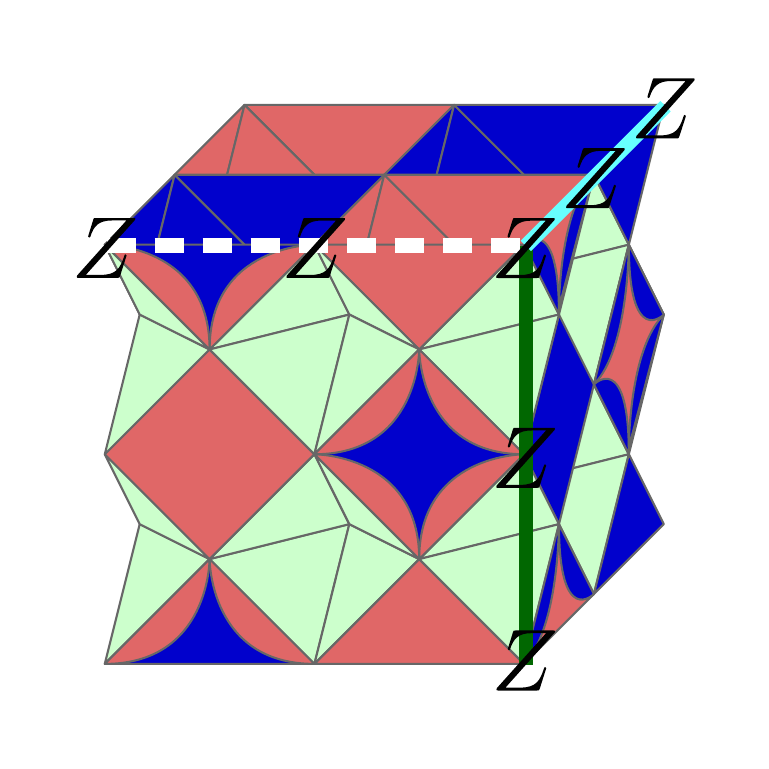}
    \caption{\label{fig:canon log ops} Canonical $\overline{Z}_{r}$ (dashed white line), $\overline{Z}_{g}$ (continuous green (dark grey) line) and $\overline{Z}_{b}$ (continuous blue (light grey) line) operators. The canonical $\overline{X}_{c}$ operators act on every qubit on one of the $c$-boundaries.}
\end{figure}

\subsection{\label{subsec:other rectified} Other rectified picture lattices}

It is natural to wonder whether the rectified cubic lattice is the only lattice which supports three 3D surface codes in the rectified picture. We say a lattice supports three 3D surface codes in the rectified picture if we can partition the cells and faces of the lattice into three sets such that we can define a valid 3D surface code for each set (with $X$ stabilizers associated with cells and $Z$ stabilizers associated with faces). In the 2D case, there are many lattices which support two surface codes in the rotated picture. Indeed, any four-valent lattice will work~\cite{Anderson2013Homological}. 

For 3D lattices, the situation is more complex. To make the analysis easier, we consider rectified lattices without boundaries. To support three 3D surface codes, a rectified picture lattice must satisfy the following conditions:
\begin{enumerate}
    \item The cells must be 3-colourable.
    \item Each vertex must be part of exactly two cells of each colour.
    \item Each vertex must be part of three or more faces of each colour.
    \item All $c$-cells and faces which are not part of $c$-cells must have an even number of vertices in common.
\end{enumerate}
Condition one allows us to assign colours to the cells and faces in a consistent way. We assign each face the colours of the two cells of which it is a member. As with rectified cubic codes, we assign each surface code, $\mathcal{SC}_{c}$, a colour. In $\mathcal{SC}_{c}$, we associate $X$ stabilizers with $c$-cells and $Z$ stabilizers with $c'c''$-faces. Condition two ensures that each qubit is acted upon non-trivially by exactly two $X$ stabilizers in each code. This is necessary because in the Kitaev picture (primal lattice) qubits are associated with edges and $X$ stabilizers with vertices. Condition three ensures that each qubit is acted upon non-trivially by three or more $Z$ stabilizers in each code. This condition is necessary to ensure that the $m$ quasiparticles are 1-D objects, as required in 3D surface codes. Finally, condition four ensures that the $X$ and $Z$ stabilizers in each code commute. In addition, we note that condition four implies Lemma~\ref{lem:overlap}. This means that as long as the three 3D surface codes have canonical logical operators which overlap as described in Figure~\ref{fig:x overlaps}, they will have a transversal $CCZ$ gate. The only semi-regular (vertex-transitive) 3D lattice we have found which satisfies the above conditions is the rectified cubic lattice. However, it is likely that other less regular lattices exist which satisfy the conditions. 

If we relax condition one, we can find regular rectified picture lattices which support more than three 3D surface codes. Instead of insisting on 3-colourability, we allow the cells of the lattice to be 4-colourable. For example, consider the cubic lattice. We can colour the cells of this lattice with four colours such that each cube has the same colour as the cubes with which it shares exactly one vertex (see Figure~\ref{fig:cubic four col}). With this colouring, the cubic lattice supports four 3D surface codes. We choose the four colours $\{r,g,b,y\}$. The four codes have the following stabilizer groups:
\begin{center}
    \begin{tabular}{||c | c | c||} 
    \hline
    Code & X stabilizers & Z stabilizers \\ [0.5ex] 
    \hline
    $\mathcal{SC}_{r}$ & $r$-cuboctahedra & $bg$-faces, $by$-faces and $gy$-faces \\ 
    \hline
    $\mathcal{SC}_{g}$ & $g$-octahedra & $rb$-faces, $ry$-faces and $by$-faces \\
    \hline
    $\mathcal{SC}_{b}$ & $b$-octahedra & $rg$-faces, $ry$-faces and $gy$-faces \\
    \hline
    $\mathcal{SC}_{y}$ & $y$-cuboctahedra & $rb$-faces, $rg$-faces and $bg$-faces  \\ [1ex]
    \hline
   \end{tabular}
\end{center}
The idea of defining a 3D surface code on the cubic lattice in this way is due to Kubica~\cite{KubicaPrivate}. However, he did not consider multiple surface codes defined on the same lattice. We have not constructed a family of codes supported on cubic lattices with boundaries, but this may be possible. 
\begin{figure}
    \includegraphics[width=0.5\columnwidth]{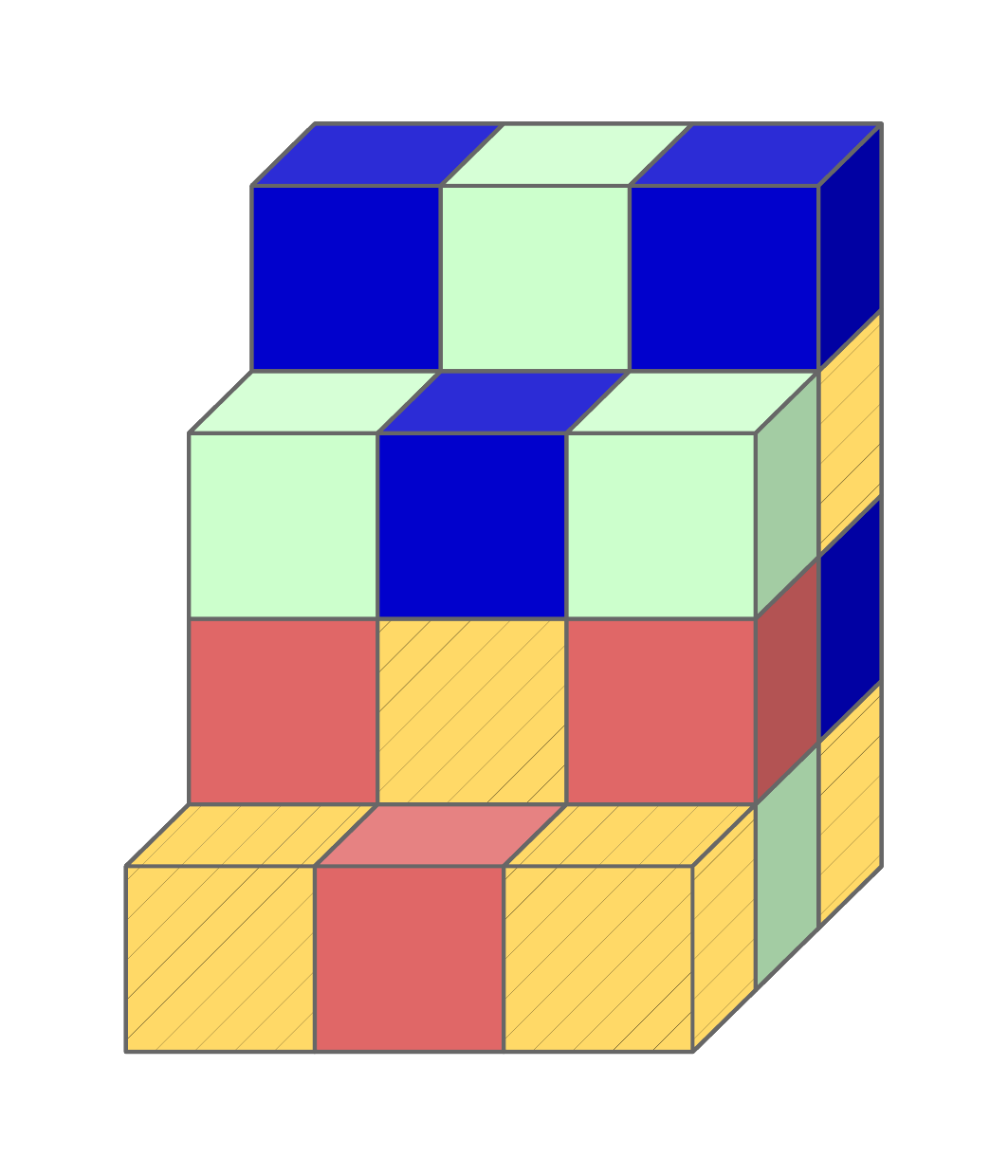}
    \caption{\label{fig:cubic four col} A cubic lattice coloured with four colours (blue (dark grey), red (medium grey), green (light grey) and yellow (hatched)). Cubes which share exactly one vertex have the same colour.}
\end{figure}

Remarkably, the cubic lattice surface codes we defined above are a gauge choice of the 3D Bacon-Shor code~\cite{BaconShor3D}, a well-know subsystem code~\cite{PoulinSubsystemCodes}. A similar result is widely known in the 2D case, see \emph{e.g}.~\cite{LiCompassCodes}. Subsystem codes are quantum error-correcting codes where the encoded qubits separate into two sets: gauge qubits and logical qubits. We only use the logical qubits to encode information, but the gauge qubits give subsystem codes additional structure which is not present in stabilizer codes. A subsystem code is defined by its gauge group $\mathcal{G}$, a subgroup of the Pauli group. The stabilizer group of the subsystem code is the centre of the gauge group, $\mathcal{S}=\mathcal{Z}(\mathcal{G})$. The non-trivial logical operators of a subsystem code are the elements of the Pauli group which commute with all the stabilizers but are not in the gauge group. In the 3D Bacon-Shor code, we place qubits on the vertices of a cubic lattice. The gauge group $\mathcal{G}$ is generated by $XX$ and $ZZ$ operators associated with the edges of the lattice. The $X$-type gauge generators are associated with edges in the $i$ and $j$ directions. Similarly, the $Z$-type gauge generators are associated with edges in the $j$ and $k$ directions. The stabilizer group contains `nearest-plane' operators. That is, the $X$-type stabilizers consist of $X$ operators acting on all the qubits in two $jk$ planes which are next to each other in the $i$ direction. Similarly, the $Z$-type stabilizers consist of $Z$ operators acting on all the qubits in two $ij$ planes which are next to each other in the $k$ direction. 

A stabilizer code defined by the stabilizer group $\mathcal{S}$ is a gauge choice of a subsystem code defined by the gauge group $\mathcal{G}_{1}$ and stabilizer group $\mathcal{S}_{1} $ if the following inclusions hold~\cite{Paetznick2013Universal,Bombin2015GaugeCC}:
\begin{equation} 
    \mathcal{S}_{1}\subseteq\mathcal{S}\subseteq\mathcal{G}_{1}. 
\end{equation}    
Consider $\mathcal{SC}_{r}$ as defined above. The $X$ stabilizers of $\mathcal{SC}_{r}$ are associated with $r$-cubes. Clearly, we can construct these cube operators from $X$ gauge operators associated with the edges in the $i$ and $j$ directions. Similarly, we can construct the $Z$ stabilizers of $\mathcal{SC}_{r}$ ($bg$, $by$ and $gy$-faces) from $Z$ gauge operators associated with the edges in the $j$ and $k$ directions. In addition, the stabilizer generators of the 3D Bacon-Shor code can be constructed from the stabilizer generators of $\mathcal{SC}_{r}$. We can construct any $X$ `nearest plane' operator from $X$ $r$-cube operators and we can construct any $Z$ `nearest plane operator' from $Z$ $bg$, $by$ and $gy$-face operators. The same is true for all the other 3D surface codes defined above by symmetry. Therefore, 3D surface codes defined on the cubic lattice (in the rectified picture) are particular gauge choices of the 3D Bacon-Shor code. 

\section{\label{sec:concat} Concatenation Transformation}

In this section, we show how to transform three 3D surface codes into a 3D color code using code concatenation. Color codes are a family of topological codes introduced by Bomb\'{i}n and Martin-Delgado~\cite{Bombin2006Topological,Bombin2007Topological}. 3D color codes are defined on weakly four-valent, 4-colourable lattices. In a weakly four-valent lattice, all vertices are four-valent except for vertices on the boundaries. In a 3D color code, we place qubits on the vertices of the lattice, we associate $X$ stabilizers with the cells of the lattice and we associate $Z$ stabilizers with the faces of the lattice. This makes 3D color codes very similar to 3D surface codes in the rectified picture. In fact, 3D color codes and stacks of three 3D surface codes are equivalent up to local Clifford unitaries, as shown by Kubica \emph{et al}.~\cite{Kubica2015Unfolding}. This result is a special case of their more general result which states that $d$ copies of a $D$-dimensional surface code are local Clifford equivalent to a single $D$-dimensional color code. The surprisingly close relationship between color codes and surface code has also been explored in a number of other works~\cite{Bombin2012Universal,Delfosse2014Decoding,Bombin2014Structure,Aloshious2016Projecting,Aloshious2018Local}.  

Criger and Terhal gave an explicit construction of the local Clifford unitaries required to transform two 2D surface codes into a single 2D color code~\cite{Criger2016Noise}. Their construction is remarkably simple - it consists of encoding pairs of qubits (one from each 2D surface code) in the [[4,2,2]] error detecting code. This code can be viewed as a 2D color code defined on a single square. It has two stabilizers $X^{\otimes 4}$ and $Z^{\otimes 4}$ and weight two logical operators supported on the sides of the square. The [[4,2,2]] code has a transversal $CZ$ gate implemented using $S=diag(1,i)$ and $S^{\dagger}$ gates. We can generalize this code concatenation transformation to 3D. Instead of a [[4,2,2]] code we use an [[8,3,2]] code. We can view this code as as a small 3D color code defined on a cube (as shown in Figure~\ref{fig:832 code}a). It has an $X$ stabilizer acting on all of the qubits and $Z$ stabilizers associated with the faces of the cube. Only four of the $Z$ face stabilizers are independent so this code has three encoded qubits. Logical $\overline{X}$ operators are membranes of $X$ operators which act on four qubits on the same face (opposite faces support $\overline{X}$ operators which act on the same encoded qubit). $\overline{Z}$ operators are strings of $Z$ operators that act on the qubits at the endpoints of edges linking the faces which support the corresponding $\overline{X}$ operators. The vertices of a cube are two-colourable, \emph{i.e}.\ we can assign each vertex a colour such that no vertices which share an edge have the same colour. We can implement a transversal $CCZ$ in the [[8,3,2]] code by applying $T=diag(1,e^{i\pi/4})$ gates to qubits on vertices of one colour and $T^{\dagger}$ gates to the qubits on the vertices of the other colour. This fact can be verified by computing the action of $T$ and $T^{\dagger}$ on the codeword kets. 

\begin{figure}
    \centering
    \begin{minipage}{0.45\columnwidth}
        \centering
        \topinset{\bfseries{a)}}{\includegraphics[width=\linewidth]{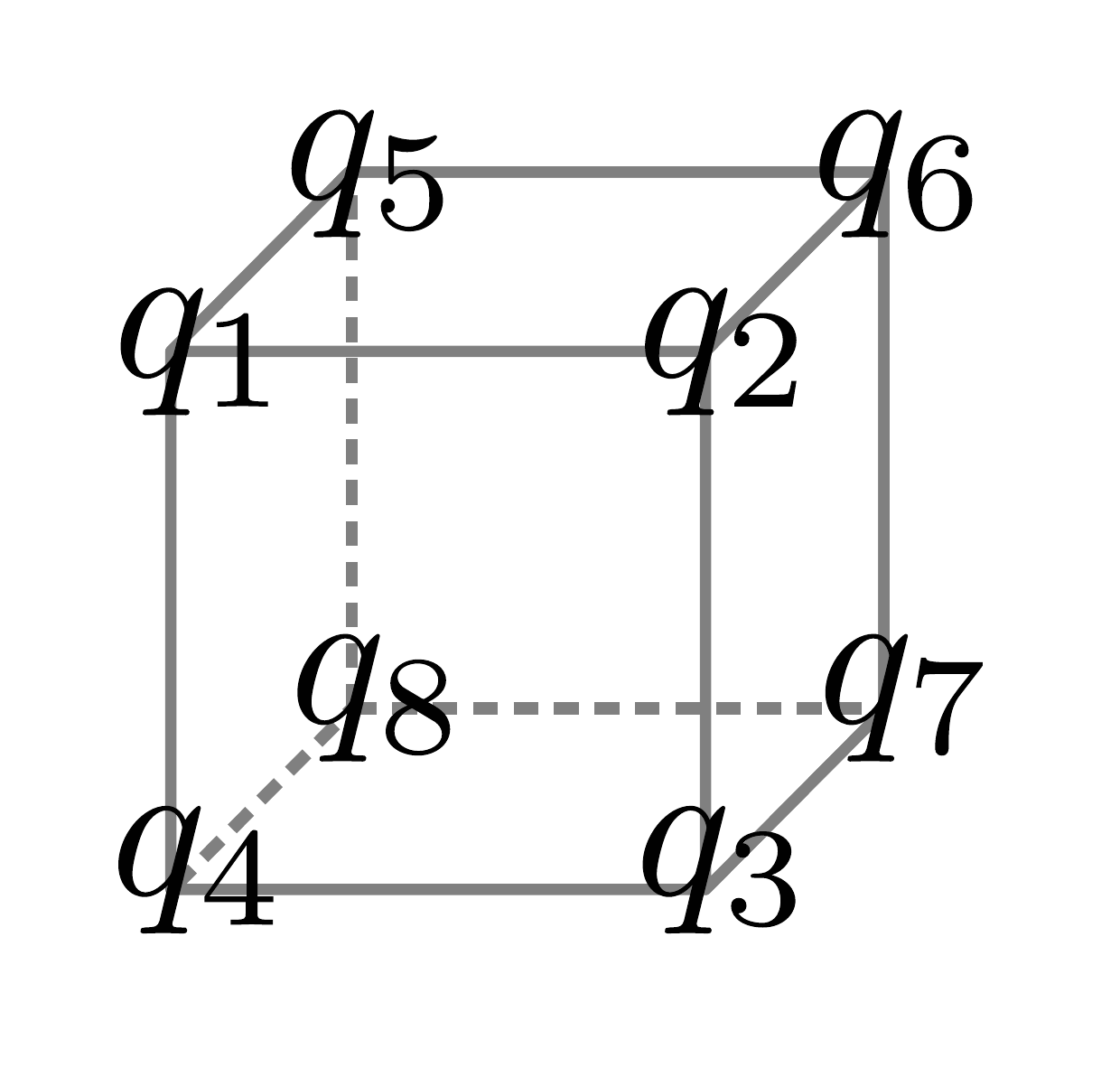}}{-0.2cm}{0cm}
    \end{minipage}
    \centering
    \begin{minipage}{0.4\columnwidth}
        \centering
        \topinset{\bfseries{b)}}{\includegraphics[width=\linewidth]{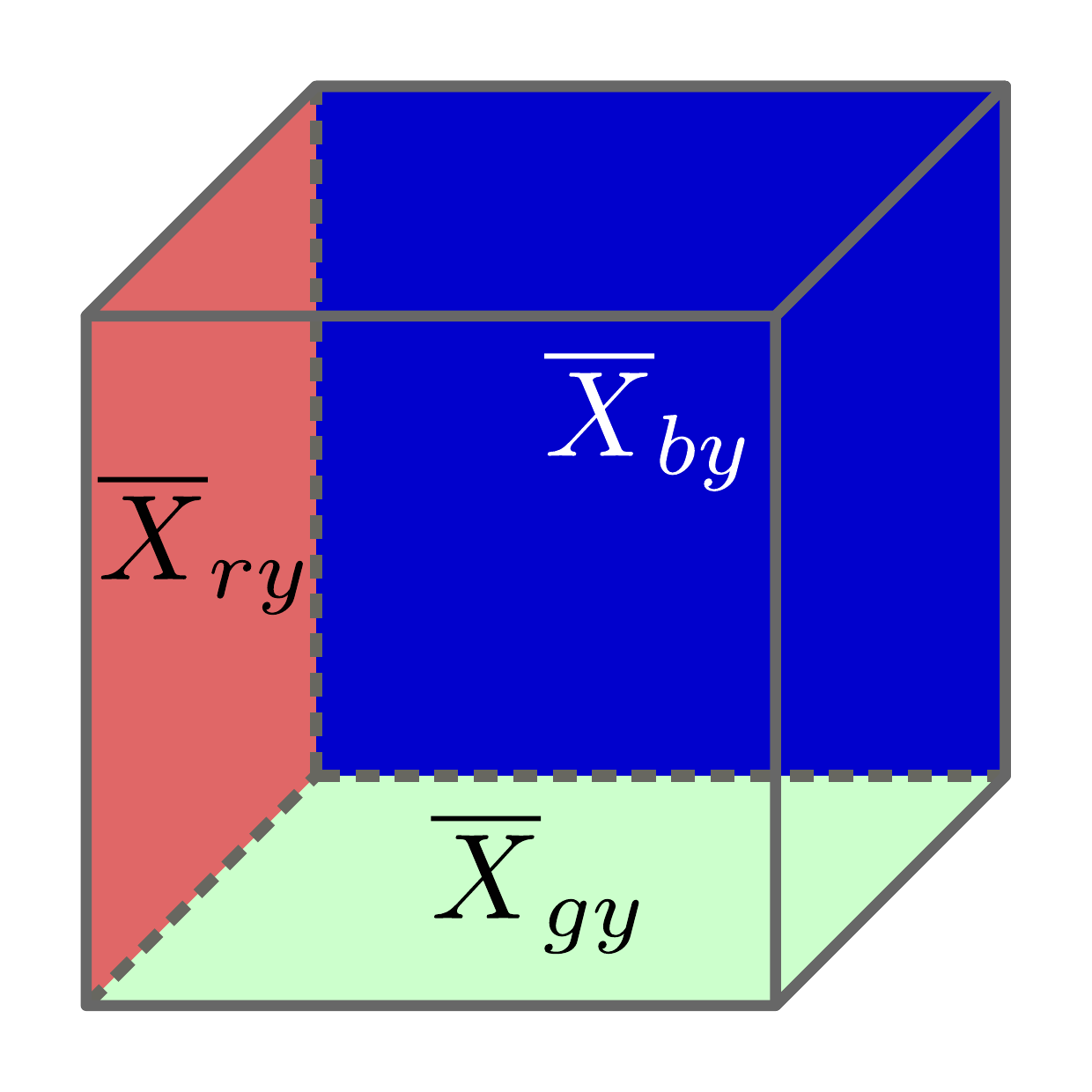}}{-0.35cm}{0cm}
    \end{minipage}
    \centering
    \begin{minipage}{0.65\columnwidth}
        \centering
        \topinset{\bfseries{c)}}{
            \includegraphics[width=\linewidth]{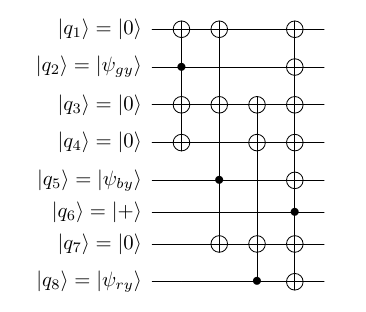}
        }{-0.6cm}{0cm}
    \end{minipage}
    \caption{\label{fig:832 code} The [[8,3,2]] color code. 
    \textbf{a)} Our labelling of the qubits. The $X$ stabilizer acts on all the qubits (it is a cell operator) and the $Z$ stabilizers act on qubits which are members of the same face. 
    \textbf{b)} In a larger 3D color code, the [[8,3,2]] cube would have an assigned colour (say $y$) and its faces would be 3-colourable ($c\in\{ry,by,gy\}$). $\overline{X}_{cy}$ is supported either of the $cy$-faces (opposite faces have the same colour) and $\overline{Z}_{cy}$ is supported on an edge that links the $cy$-faces.
    \textbf{c)} The encoding circuit for the [[8,3,2]] code. Encoded $\overline{X}_{cy}$ operators (shown in \textbf{b}) act on the encoded qubit $\ket{\overline{\psi}_{cy}}$. We derived this circuit using the method given in~\cite{Gottesman1997Stabilizer}.
    }
\end{figure}

In a 3D color code, we assign faces the colours of the cells they are members of. For example, a face which is a member of a $c$-cell and a $c'$-cell is a $cc'$-face. Due to the 4-colourability of the color code lattice, each cell's faces are 3-colourable. Consider a color code lattice where cells assigned colours from the set $\{r,g,b,y\}$. We can view the [[8,3,2]] code as a cell of this lattice. Assume that it is a $y$-cell. Then, its faces are coloured $ry$, $by$ and $gy$. We use these colours to index the logical operators of the [[8,3,2]] code. That is, the logical $\overline{X}$ operators which act on the $cy$-faces are denoted by $\overline{X}_{cy}$. These operators are shown in Figure~\ref{fig:832 code}b. We denote the corresponding $\overline{Z}$ operators as $\overline{Z}_{cy}$. 

\begin{figure}[ht]
    \centering
    \begin{minipage}{0.4\columnwidth}
        \centering
        \topinset{\bfseries{a)}}{\includegraphics[width=\linewidth]{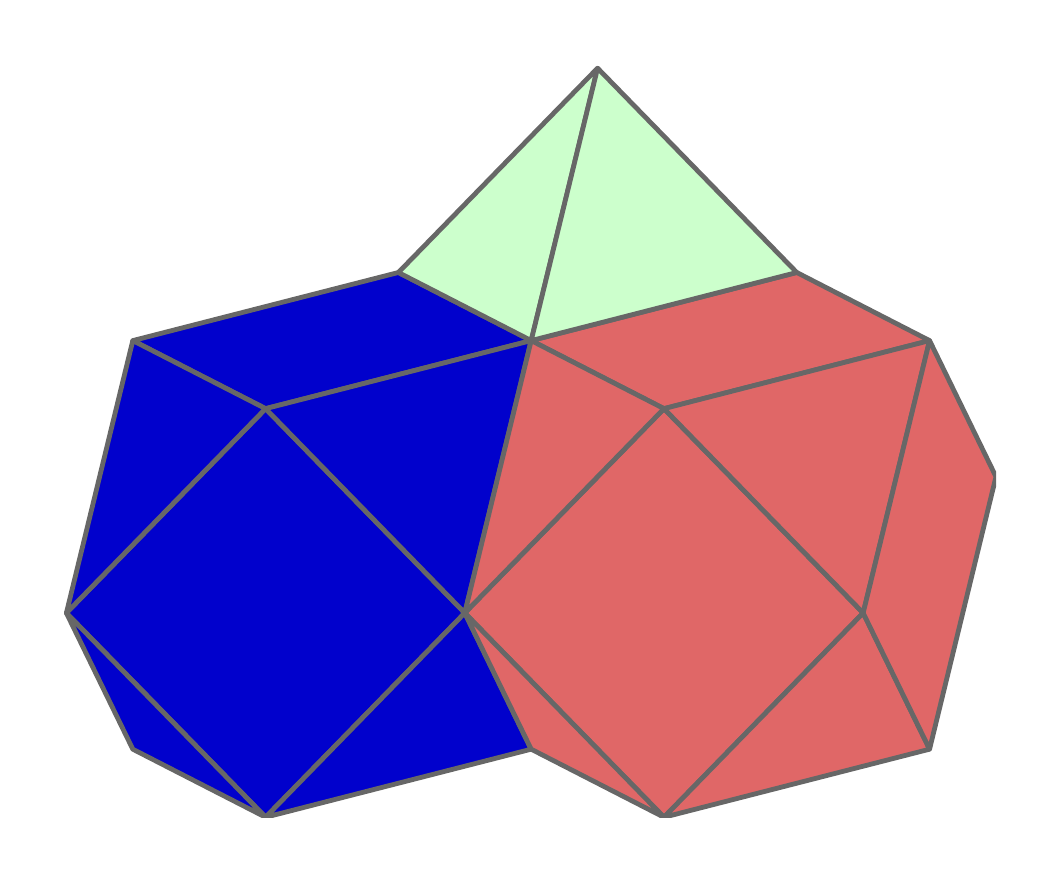}}{-0.35cm}{0cm}
    \end{minipage}
    \begin{minipage}{0.4\columnwidth}
        \centering
        \topinset{\bfseries{b)}}{\includegraphics[width=\linewidth]{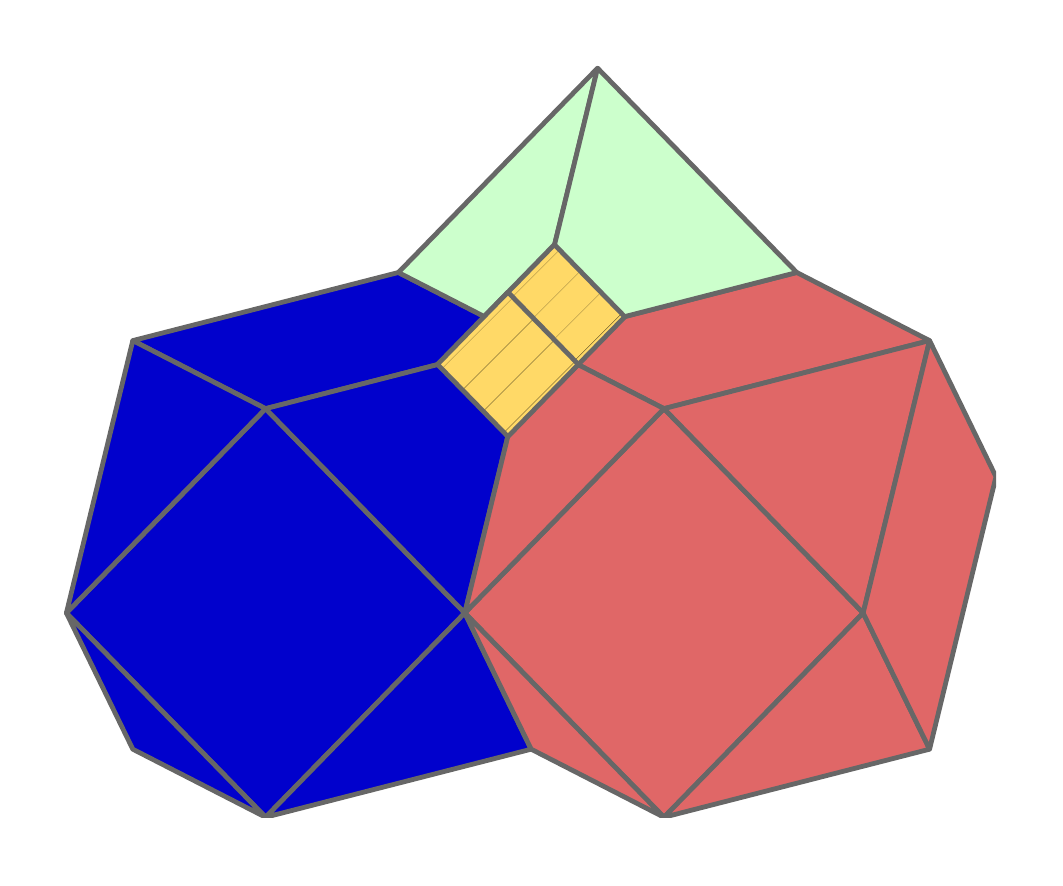}}{-0.35cm}{0cm}
    \end{minipage}
    \caption{\label{fig:832 vertex} An [[8,3,2]] concatenation transformation of a single vertex in a stack of 3D surface codes. 
    \textbf{a)} The initial rectified cubic lattice. 
    \textbf{b)} We encode the three qubits at the vertex where the three different cells meet in an [[8,3,2]] color code. This corresponds to replacing the vertex with a cube (yellow (hatched) cell).}
\end{figure}

We can now detail the concatenation transformation which maps a stack of three 3D surface codes to a single 3D color code. Consider a rectified cubic code stack with code distance $d$. To transform the three codes in the stack, we encode the three qubits at every vertex in [[8,3,2]] codes. An encoding circuit for the [[8,3,2]] code is shown in Figure~\ref{fig:832 code}c. Figure~\ref{fig:832 vertex} shows the [[8,3,2]] concatenation transformation applied to a single vertex. Applied to a whole lattice, concatenation with the [[8,3,2]] code transforms cuboctahedra into truncated cuboctahedra, octahedra into truncated octahedra and vertices into cubes. Globally this transforms the rectified cubic lattice into a cantitruncated cubic lattice. Two truncated cuboctahedra, one truncated octahedron and one cube meet at each vertex of a cantitruncated cubic lattice. Figure~\ref{fig:832 lattice} shows how a $d=2$ rectified cubic lattice transforms under the [[8,3,2]] concatenation transformation. 

\begin{figure}[ht]
    \centering
    \begin{minipage}{0.4\columnwidth}
        \centering
        \topinset{\bfseries{a)}}{\includegraphics[width=\linewidth]{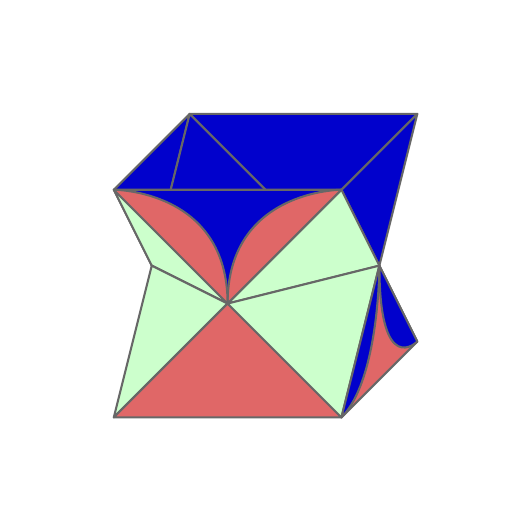}}{-.2cm}{0cm}
    \end{minipage}
    \begin{minipage}{0.55\columnwidth}
        \centering
        \topinset{\bfseries{b)}}{\includegraphics[width=\linewidth]{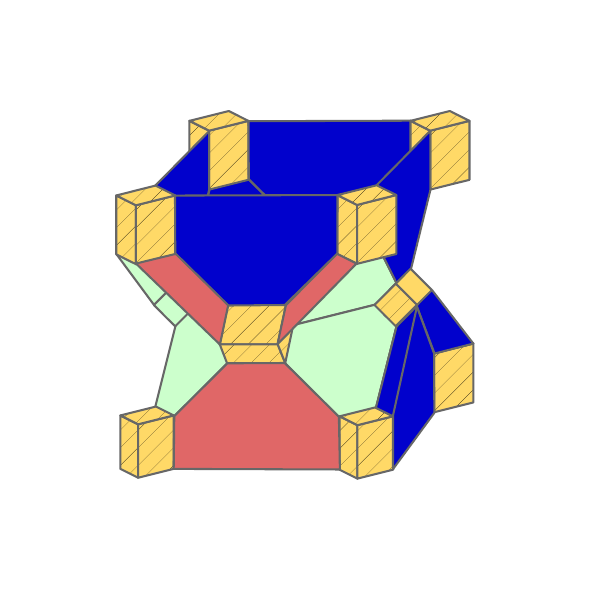}}{0.5cm}{0cm}
    \end{minipage}
    \caption{\label{fig:832 lattice} Transforming a stack of three 3D surface codes into a single 3D color code by concatenating with the [[8,3,2]] color code. Each vertex the $d=2$ rectified cubic lattice (\textbf{a}) is transformed as shown in Figure~\ref{fig:832 vertex}. This transforms the rectified cubic lattice into a cantitruncated cubic lattice (\textbf{b}). This lattice supports a $d=4$ color code with three encoded qubits. The top and bottom boundaries of the surface code stack are $g$-boundaries, the left and right boundaries are $r$-boundaries, and the front and back boundaries are $b$-boundaries. In a color code, a $c$-boundary is a boundary which has no $c$-cells adjacent to it. By inspecting Figure \textbf{b}, we see that $c$-boundaries in the surface code stack become $c$-boundaries in the color code.}
\end{figure}

The colours we assigned to the encoded qubits of the [[8,3,2]] codes tell us how to encode the three qubits at every vertex of the rectified cubic lattice. We encode the physical qubits from $\mathcal{SC}_{c}$ as the $cy$-qubits of the [[8,3,2]] codes (see Figure~\ref{fig:832 code}c). This ensures that $\mathcal{SC}_{c}$ $X$ ($Z$) stabilizers associated with $c$-cells ($cc'$-faces) are mapped to $X$ ($Z$) stabilizers associated with $c$-cells ($cc'$-faces) in the color code. In each 3D surface code we have $n$ qubits and $n-1$ independent stabilizer generators. In the color code we have $8n$ qubits and we inherit $3(n-1)$ stabilizer generators. We also have five independent stabilizer generators for each cube (one $X$ stabilizer and four $Z$ stabilizers). So in total we have $5n+3n-3=8n-3$ independent stabilizer generators in the 3D color code. The 3D color code therefore encodes three logical qubits. The color code inherits the boundary structure of the stack of 3D surface codes. In a color code, a boundary has the colour $c$, if no $c$-cells are present on it. As shown in Figure~\ref{fig:832 lattice}, the $c$-boundaries of the rectified cubic lattice become $c$-boundaries in the color code. 

As with surface codes, we interpret unsatisfied color code stabilizers as quasiparticles. For each colour $c$ in a 3D color code, we have quasiparticles $e_{c}$ and $m_{c}$. In the stack of surface codes, each code $\mathcal{SC}_{c}$ has quasiparticles $e_{c}$ and $m_{c}$. The quasiparticles of the three surface codes are mapped directly to quasiparticles in the color code. For any colour $c\in\{r,g,b\}$, the $e_{c}$ ($m_{c}$) quasiparticles in our three 3D surface codes are mapped to $e_{c}$ ($m_{c}$) quasiparticles in the color code because $c$-cell ($cc'$-face) stabilizers in the surface codes are mapped to $c$-cell ($cc'$-face) stabilizers in the color code. This leaves the $y$ quasiparticles in the color code unaccounted for. However, this is not important as the $y$ quasiparticles are not independent. We can always construct $y$ quasiparticles from combinations of $r$, $g$ and $b$ quasiparticles~\cite{Bombin2007Topological}. 

The logical operators of the color code have the same structure as the logical operators of the three surface codes. In the color code, $\overline{Z}_{c}$ operators are strings of $Z$ operators from one $c$-boundary to the other and $\overline{X}_{c}$ operators are membranes of $X$ operators with boundaries that span the $c'$ and $c''$-boundaries. As the concatenation transformation maps the $c$-boundaries of the rectified cubic codes to $c$-boundaries in the color code, the structure of the logical operators is preserved by the mapping. 

\section{\label{sec:gates} A Universal Gate Set in 3D Surface Codes}

In this section, we prove that $CCZ$ and $CZ$ are transversal in rectified cubic codes and we show how to implement a universal gate set in these codes. We note that $CZ$ is also transversal in 2D surface codes. This fact can be easily understood in the rotated picture, as we explain in Appendix~\ref{app:2D cz}.

An important concept in our proofs is the overlap of logical operators (including stabilizers). Given two or three logical operators, each of which acts on a different code in a rectified cubic stack, we define the overlap of these operators as the vertices where all the operators act non-trivially. Before proceeding to the main proofs, we need the following Lemma about rectified cubic codes.
\begin{lemma}
    The overlap of any two $X$ stabilizers from two different codes in a rectified cubic code stack is equal to the non-trivial support of a $Z$ stabilizer from the third code.
    \label{lem:overlap}
\end{lemma}
In other words, the set of vertices at which both $X$ stabilizers act non-trivially are equal to the support of some $Z$ stabilizer in the third code.
\begin{proof}
We initially restrict our attention to the bulk of the lattice. Let us consider $X$ stabilizer generators from $\mathcal{SC}_{r}$ ($r$-cells) and $\mathcal{SC}_{g}$ ($g$-cells). We denote the $X$ and $Z$ stabilizers of $\mathcal{SC}_{c}$ as $S_{c}^{x}$ and $S_{c}^{z}$, respectively. Clearly, $S_{r}^{x}$ generators and $S_{g}^{x}$ generators overlap on $rg$-faces ($S_{b}^{z}$ operators) in the bulk. This is also true for the other two colour combinations.

On the boundaries, $S_{r}^{x}$ generators and $S_{b}^{x}$ operators overlap on $rb$-faces. This can be seen by inspecting \emph{e.g}.\ Figure~\ref{fig:extra stabs}. Some $S_{r}^{x}$ generators and $S_{g}^{x}$ generators overlap on edges. However, in all these cases, a $S_{b}^{z}$ operator is supported on the overlap edge. An example of this is highlighted in Figure~\ref{fig:stab overlap}. Similarly, $S_{b}^{x}$ and $S_{g}^{x}$ generators on the boundaries can overlap on edges. But all these edges have an associated $S_{r}^{z}$ operator. 

\begin{figure}
    \includegraphics[width=0.6\columnwidth]{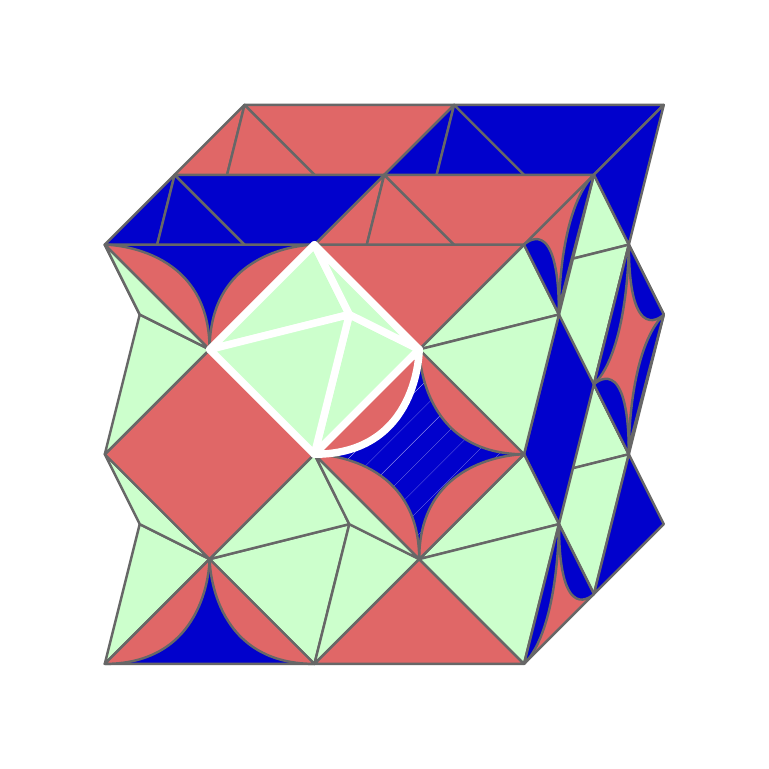}
    \caption{\label{fig:stab overlap} The overlap of a $S_{g}^{x}$ generator (green (light grey) octahedron with white edges) and a $S_{r}^{x}$ generator (hatched blue face) is equal to an edge (red (medium grey) circular segment with white edges). This edge has an associated $S_{b}^{z}$ operator, as explained in Figure~\ref{fig:extra stabs}.}
\end{figure}

We have shown that all pairs of $X$ stabilizer generators from two different codes in the stack overlap on faces or edges which support $Z$ stabilizers in the third code. As every $X$ stabilizer is a product of stabilizer generators, any pair of $X$ stabilizers from two different codes have overlap equal to the support of a $Z$ stabilizer from the third code.
\end{proof}

\subsection{\label{subsec:ccz} Transversal CCZ}

We now prove that $CCZ$ is transversal for stacks rectified cubic codes. We first write the surface code kets in a form inspired by a proof in~\cite{Paetznick2013Universal}. Let $H_{c}^{x}$ be the (classical) parity check matrix of the $X$ stabilizers of $\mathcal{SC}_{c}$. That is, $H_{c}^{x}$ is an $m$ by $n$ binary matrix with $m$ equal to the number of $X$ stabilizer generators in $\mathcal{SC}_{c}$ and $n$ equal to the number of physical qubits in the code. Each row of $H_{c}^{x}$ has a 1 at column $j$ if the stabilizer generator corresponding to that row acts non-trivially on qubit $q_{j}$. If the stabilizer generator acts trivially then the entry is equal to zero. Now let $G_{c}^{0}$ be the linear span the rows of $H_{c}^{x}$. For each code, we choose a canonical $\overline{X}_{c}$ operator which acts on one of the $c$-boundaries of the lattice. Let $X_{c}$ be an $n$-bit binary vector describing the support of $\overline{X}_{c}$. That is, $X_{c}$ has a one at position $j$ if $\overline{X}_{c}$ acts non-trivially on qubit $q_{j}$, with all other entries in $X_{c}$ equal to zero. Let $G_{c}^{1}$ be the coset $\{X_{c}+g:g\in G_{c}^{0}\}$. With these definitions we can write the encoded state of $\mathcal{SC}_{c}$ as follows:
\begin{equation}
    \ket{\overline{\alpha}}_{c}=
    \frac{1}{\sqrt{|G_{c}^{\alpha}|}}\sum_{g\in G_{c}^{\alpha}}\ket{g}_{c},
    \label{eq:ket form}
\end{equation}
where $|G_{c}^{\alpha}|$ is the number of elements in $G_{c}^{\alpha}$ and $\alpha\in\{0,1\}$.

To show that $CCZ$ is transversal for stacked 3D surface codes we need the following lemma.
\begin{lemma}
    Given a finite set of $k$ binary vectors $\{a_{j}\}$ with the same length, the parity of their sum is equal to the sum of their parities. 
    \label{lem:parity}
\end{lemma}
This lemma is easy to prove. For completeness, we include a proof in Appendix~\ref{app:lemma}.

\begin{theorem}
    $CCZ$ is tranversal in rectified cubic codes.
    \label{thm:ccz}
\end{theorem}

\begin{proof}
Define $\overline{CCZ}=CCZ^{\otimes n}$, where each $CCZ$ gate acts on the three qubits (one per code) at one of the $n$ vertices of the lattice. We consider the initial state
\begin{equation}
    \ket{\overline{\alpha\beta\gamma}}_{rgb}=
    \sum_{t\in G_{r}^{\alpha}u\in G_{g}^{\beta}v\in G_{b}^{\gamma}}
    \ket{t}_{r}\ket{u}_{g}\ket{v}_{b},
    \label{eq:CCZ start}
\end{equation}
where $\alpha,\beta,\gamma\in\{0,1\}$. We have omitted the global normalization factor. Now, we apply $\overline{CCZ}$ to $\ket{\overline{\alpha\beta\gamma}}_{rgb}$:
\begin{equation}
    \begin{split}
        \overline{CCZ}\ket{\overline{\alpha\beta\gamma}}_{rgb}&=
        \sum_{t\in G_{r}^{\alpha}u\in G_{g}^{\beta}v\in G_{b}^{\gamma}}
        CCZ^{\otimes n}\ket{t}_{r}\ket{u}_{g}\ket{v}_{b}, \\
        &=\sum_{t\in G_{r}^{\alpha}u\in G_{g}^{\beta}v\in G_{b}^{\gamma}}
        (-1)^{|t\cdot u\cdot v|}\ket{t}_{r}\ket{u}_{g}\ket{v}_{b},
    \end{split}
    \label{eq:CCZ action}
\end{equation}
where $u\cdot v$ denotes the bitwise binary product between $u$ and $v$ and $|t|$ denotes the Hamming weight of $t$. 

We now calculate $(-1)^{|t\cdot u\cdot v|}$ for each encoded computational basis state. We can expand $t\cdot u\cdot v$ as follows:
\begin{equation}
    \begin{split}
        t\cdot u\cdot v&=(\alpha X_{r}+t')\cdot (\beta X_{g}+u')\cdot (\gamma X_{b}+v') \\
        &=\alpha \beta \gamma(X_{r}\cdot X_{g} \cdot X_{b}) +
        \alpha \beta(X_{r}\cdot X_{g}\cdot v') \\
        &+ \alpha \gamma(X_{r}\cdot X_{b}\cdot u')
        + \beta \gamma(X_{g}\cdot X_{b}\cdot t') \\
        &+ \alpha(X_{r}\cdot u'\cdot v') 
        + \beta (X_{g}\cdot t'\cdot v') \\
        &+ \gamma(X_{b}\cdot t'\cdot u') 
        + (t'\cdot u'\cdot v')
    \end{split}
    \label{eq:exponent expansion}
\end{equation}
where $t'\in G_{0}^{r}$, $u'\in G_{0}^{g}$ and $v'\in G_{0}^{b}$. 

First we consider the term $(t'\cdot u'\cdot v')$, which corresponds to the state $\ket{\overline{000}}$. We can find the Hamming weight of this term by considering the support of the stabilizers which correspond to $t'$, $u'$ and $v'$. The $t'$ vectors correspond to the $X$ stabilizers of $\mathcal{SC}_{r}$ ($S_{r}^{x}$), the $u'$ vectors correspond to the $X$ stabilizers of $\mathcal{SC}_{g}$ ($S_{g}^{x}$) and the $v'$ vectors correspond to the $X$ stabilizers of $\mathcal{SC}_{b}$ ($S_{b}^{x}$). The Hamming weight of the product $t'\cdot u' \cdot v'$ will be equal to the number of vertices in the lattice where the three $X$ stabilizers act non-trivially on the physical qubits of their respective codes. In other words, it will be equal to the overlap of the three operators. 

By Lemma~\ref{lem:overlap}, any $S_{r}^{x}$ operator and any $S_{g}^{x}$ operator have overlap equal to the support of a $S_{b}^{z}$ operator ($\mathcal{SC}_{b}$ $Z$ stabilizer). As the stabilizers of $\mathcal{SC}_{b}$ commute, the overlap of any $S_{r}^{x}$, $S_{g}^{x}$ and $S_{b}^{x}$ is always even. Hence, $|t'\cdot u' \cdot v'|=0\mod 2$ for all $t'$, $u'$ and $v'$ and $(-1)^{|t\cdot u\cdot v|}=1$ for $\ket{\overline{000}}$.

Next we consider the exponent for $\ket{\overline{001}}$ which is equal to $(X_{b}\cdot t' \cdot u')+(t'\cdot u'\cdot v')$. Thanks to Lemma~\ref{lem:parity} we only need to show that $(X_{b}\cdot t' \cdot u')$ has even Hamming weight to show that the sum has even Hamming weight. We need to calculate the overlap of the $\overline{X}_{b}$ operator on the $b$-boundary (corresponding to the $X_{b}$ vector) with any $S_{r}^{x}$ and $S_{g}^{x}$. By Lemma~\ref{lem:overlap}, any $S_{r}^{x}$ and $S_{g}^{x}$ overlap on a collection of vertices which has the same support (in terms of vertices) as a $S_{b}^{z}$ operator. Logical operators and stabilizers commute, so the overlap of $\overline{X}_{b}$ with any $S_{r}^{x}$ and $S_{g}^{x}$ is even. This implies that $|(X_{b}\cdot t' \cdot u')|=0\mod 2$ for every $t'$ and $u'$. All the other terms in Equation~\ref{eq:exponent expansion} with one $X_{c}$ term have even Hamming weight by the same argument. Therefore $(-1)^{|t\cdot u\cdot v|}=1$ for $\ket{\overline{100}}$, $\ket{\overline{010}}$ and $\ket{\overline{001}}$. 

\begin{figure}
    \includegraphics[width=0.5\columnwidth]{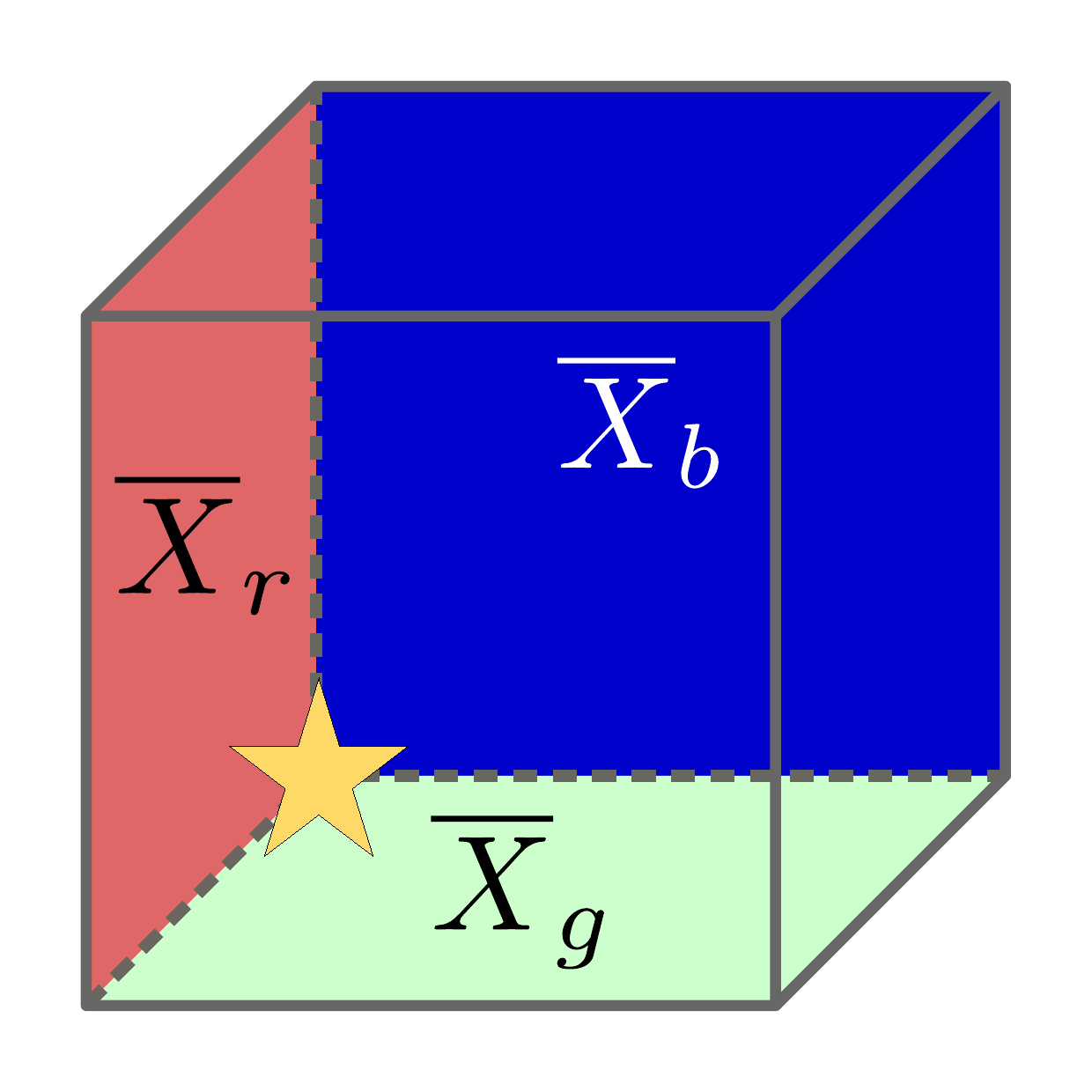}
    \caption{\label{fig:x overlaps} The overlap of the $\overline{X}_{c}$ operators which lie on the boundaries. We see that any $\overline{X}_{c}$ and $\overline{X}_{c'}$ overlap on a $\overline{Z}_{c''}$ path (a string from one $c''$-boundary to the other). The three $\overline{X}_{c}$ operators overlap at a single vertex (denoted by a star).}
\end{figure}

The next computational basis state we consider is $\ket{\overline{110}}$. The exponent for this state is $(X_{r}\cdot X_{g} \cdot v')+(X_{r}\cdot u'\cdot v')+(X_{g}\cdot t'\cdot v')+(t'\cdot u'\cdot v')$. To show that this expression has even Hamming weight we only need to show that $(X_{r}\cdot X_{g} \cdot v')$ has even Hamming weight due to Lemma~\ref{lem:parity}. To find the Hamming weight of this term we need to find the overlap of $\overline{X}_{r}$, $\overline{X}_{g}$ and any $S_{b}^{x}$ operator. $\overline{X}_{r}$ has non-trivial support on an $r$-boundary and $\overline{X}_{g}$ has non-trivial support on a $g$-boundary. These two operators overlap on a line where the $r$-boundary and the $g$-boundary meet (shown in Figure~\ref{fig:x overlaps}). This line is a string from one $b$-boundary to the other $b$-boundary \emph{i.e}.\ it has the same support as a $\overline{Z}_{b}$ operator. Logical operators and stabilizers commute so $\overline{X}_{r}$, $\overline{X}_{g}$ and any $S_{b}^{x}$ have even overlap. This proves that $|(X_{r}\cdot X_{g} \cdot v')|=0\mod 2$ for all $v'$. All the other terms in the Equation~\ref{eq:exponent expansion} expansion with two $X_{c}$ terms have even Hamming weight by the same argument. Hence, $(-1)^{|t\cdot u\cdot v|}=1$ for $\ket{\overline{110}}$, $\ket{\overline{101}}$ and $\ket{\overline{011}}$. 

Finally, for the state $\ket{\overline{111}}$ we must consider the entire expansion in Equation~\ref{eq:exponent expansion}. Due to the previous calculations in this proof and Lemma~\ref{lem:parity}, the parity of this exponent is determined by $(X_{r}\cdot X_{g}\cdot X_{b})$. This term has Hamming weight equal to the number of lattice collisions between $\overline{X}_{r}$, $\overline{X}_{g}$ and $\overline{X}_{b}$. As these three operators are defined on $r$, $g$ and $b$-boundaries respectively, they have a single lattice collision on one corner of the lattice (shown in Figure~\ref{fig:x overlaps}). Therefore $|(X_{r}\cdot X_{g} \cdot X_{b})|=1$ which implies that $(-1)^{|t\cdot u\cdot v|}=-1$ for $\ket{\overline{111}}$.

We have shown that $\overline{CCZ}$ has has the correct action on the computational basis states, namely:
\begin{equation}
\overline{CCZ}\ket{\overline{\alpha\beta\gamma}}=
    \begin{cases}
        -\ket{\overline{\alpha\beta\gamma}}\quad \alpha=\beta=\gamma=1, \\
        \quad\ket{\overline{\alpha\beta\gamma}}\quad else.
    \end{cases}
    \label{eq:Transversal CCZ}
\end{equation}
\end{proof}

\subsection{\label{app:cz} Transversal CZ}

The transversality of $CZ$ in stacked 3D surface codes follows from the structure of the $CCZ$ and $X$ operators. Consider three codes, each encoding one logical qubit, labelled with the labels $i$, $j$ and $k$. We assume that $\overline{CCZ}_{ijk}$ is a transversal gate acting as a tensor product of $CCZ$ gates at the level of the physical qubits. In addition, we assume that each $\overline{X}$ gate acts as a tensor product of $X$ gates at the level of the physical qubits. The group commutator of two operators $A$ and $B$ is defined as $K[A, B]=ABA^{\dagger}B^{\dagger}$. One can easily verify that $K[CCZ_{ijk}, X_{k}]=CZ_{ij}$. Therefore, we can implement a transversal $\overline{CZ}_{ij}$ gate by applying the sequence of logical operators $K[\overline{CCZ}_{ijk}, \overline{X}_{k}]$. If we think at the level of the physical qubits, this operator simplifies. All triples of qubits outside the support of $\overline{X}_{k}$ are acted upon by $CCZ_{ijk}CCZ^{\dagger}_{ijk}=I$ and triples of qubits in the support of $\overline{X}_{k}$ are acted upon by $K[CCZ_{ijk}, X_{k}]=CZ_{ij}$. Therefore, we can implement a transversal logical $\overline{CZ}_{ij}$ by applying $CZ$ gates at the level of the physical qubits.

In the context of our stacked 3D surface codes, the above argument implies that we can implement a logical $\overline{CZ}_{cc'}$ gate by applying $CZ$ gates to the pairs of physical qubits in $\mathcal{SC}_{c}$ and $\mathcal{SC}_{c'}$ at the vertices of one of the $c''$-boundaries (our canonical $X_{c''}$ operators are supported on the $c''$-boundaries). 

\subsection{\label{subsec:universal} Completing a universal set of gates}

To achieve universal quantum computing with a $\overline{CCZ}$ gate we only need a Hadamard gate ($H=(X+Z)/\sqrt{2}$)~\cite{Shi2003Both}. The $H$ gate is not transversal in 3D surface codes, but we can still implement it using the teleportation circuit~\cite{Zhou2000Methodology} shown in Figure~\ref{fig:teleH}. Therefore, $CCZ$ is universal if we have access to measurement and state preparation in the $X$ and $Z$ bases~\cite{Yoder2016Universal}. As long as we have access to a decoder with a threshold, we can prepare states in the $X$ basis or the $Z$ basis and we can measure qubits in the $X$ basis or the $Z$ basis. We delay discussing decoding strategies for 3D surface codes until Section~\ref{sec:arch}. We can generalize the state preparation and measurement methods used in 2D surface codes~\cite{Dennis2002Topological} to 3D surface codes. We quickly review these methods here for completeness. To measure a qubit encoded in a 3D surface code in the $Z$ basis we simply measure all of the qubits in the code in the $Z$ basis and compute the eigenvalues of all the $Z$ stabilizers. We then correct any $X$ errors implied by this syndrome using a decoder. Finally we compute the parity of a $\overline{Z}$ operator using the corrected qubit values. To measure in the $X$ basis we replace $X$ with $Z$ (and vice versa) in the procedure we have just described. To fault-tolerantly prepare a $\ket{\overline{0}}$ state we prepare each of the physical qubits in the $\ket{0}$ state. We then perform $d$ rounds of error correction (where $d$ is the code distance). To fault-tolerantly prepare $\ket{\overline{+}}$ we just replace $\ket{0}$ with $\ket{+}$ in the above procedure. 

\begin{figure}[ht]
    \includegraphics[width=0.7\linewidth]{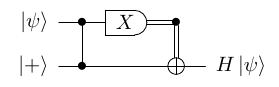}
    \caption{\label{fig:teleH} A circuit which implements a $H$ gate using state preparation, measurement and $CZ$~\protect\cite{Zhou2000Methodology}.}
\end{figure}

We can use the circuit in Figure~\ref{fig:teleH} to implement a single qubit $H$ gate in a stack of three 3D surface codes and to transfer a logical qubit between different codes in the same stack. We denote the circuit in Figure~\ref{fig:teleH} as $H_{cc'}$. This circuit takes the state $\ket{\overline{\psi}}_{c}$ to $\overline{H}\ket{\overline{\psi}}_{c'}$. Consider the initial state $\ket{\overline{\psi}}_{r}\ket{\overline{+}}_{g}\ket{\overline{+}}_{b}$. We can use sequences of $H_{cc'}$ circuits to transfer the state from one code to another or to perform a single qubit $H$ gate as follows: 
\begin{equation}
    \begin{split}
        &\ket{\overline{\psi}}_{r}\xrightarrow{H_{rg}}H\ket{\overline{\psi}}_{g}\xrightarrow{H_{gb}}\ket{\overline{\psi}}_{b}, \\
        &\ket{\overline{\psi}}_{r}\xrightarrow{H_{rg}}H\ket{\overline{\psi}}_{g}\xrightarrow{H_{gb}}\ket{\overline{\psi}}_{b}\xrightarrow{H_{br}}H\ket{\overline{\psi}}_{r}.
    \end{split}
    \label{eq:H and tele}
\end{equation}

\section{\label{sec:lattice surgery} 3D Surface Code Lattice Surgery}

We have shown how to implement a universal gate set in a single stack of three 3D surface codes. However, in a feasible architecture we also need to be able to transfer qubits between surface codes in different stacks. To accomplish this task we generalize the techniques of 2D surface code lattice surgery~\cite{Horsman2012Surface,Litinski2017Lattice,FowlerGidney2018LatticeSurgery,Litinski2018GameOfSurfaceCodes} to 3D surface codes. We note that we will reproduce some material from~\cite{Horsman2012Surface} to make our exposition clearer. Lattice surgery is a code deformation technique which allows us to merge two surface codes into a larger surface code or to split a surface code into two smaller surface codes. Lattice surgery merges and splits can be used for to transfer qubits between codes or to implement $CNOT$ gates. In related recent work, lattice surgery techniques have been extended to the Raussendorf lattice~\cite{Herr2017LatticeSurgery}, a lattice used in fault-tolerant measurement-based quantum computing~\cite{Raussendorf2006Fault}. 

There are two types of lattice surgery we can do in 3D surface codes: $X$-type and $Z$-type (corresponding to rough and smooth lattice surgery in the language of~\cite{Horsman2012Surface}). We start by presenting lattice surgery techniques for pairs of 3D surface codes before presenting a method for doing lattice surgery on a 3D surface code and a 2D surface code.

\subsection{\label{subsec:3d3dls}3D-3D lattice surgery}

We start with $X$-type lattice surgery. Consider two distance $d$ rectified cubic lattices. Each lattice supports three surface codes, $\mathcal{SC}_{c}^{(i)}$, where $c\in\{r,g,b\}$ and $i\in\{1,2\}$ indexes the two stacks. We can do an $X$-type lattice surgery merge between $\mathcal{SC}_{c}^{(1)}$ and $\mathcal{SC}_{c}^{(2)}$ by aligning $c$-boundaries of the two stacks (the rough boundaries of the two codes), preparing a layer of ancillas in the $\ket{0}$ state between the stacks and then measuring new $X$ stabilizers which join the two lattices. The product of these $X$ stabilizers is $\overline{X}_{c}^{(1)}\otimes\overline{X}_{c}^{(2)}$ so we learn this value when we perform the merge operation. There may also be new $Z$ stabilizers which we add to the stabilizer group and measure in subsequent rounds. In addition, some $Z$ stabilizers on the boundaries where the merge took place may need be modified in the new stabilizer group. The merge operation maps $\ket{\psi}_{c}\otimes\ket{\phi}_{c}\rightarrow\alpha\ket{\psi}_{c}+(-1)^{m}\beta X\ket{\psi}_{c}$, where $m$ is the outcome of the $\overline{X}_{c}^{(1)}\otimes\overline{X}_{c}^{(2)}$ measurement and $\ket{\phi}_{c}=\alpha\ket{0}+\beta\ket{1}$~\cite{Horsman2012Surface}. Any $\overline{X}$ operator for either of the two initial codes is a valid $\overline{X}$ operator for the merged code. However, to form a logical $\overline{Z}$ operator in the new code we must join logical $\overline{Z}$ operators from each of the initial codes into a single string of $Z$ operators which starts and ends at opposite $c$-boundaries. We implement an $X$-type lattice surgery split by measuring all the qubits in a layer where we want to split the lattice in the $Z$ basis. This splits the single surface code into two smaller surface codes. An $X$-type split performed on $\mathcal{SC}_{c}$ implements the following mapping: $\alpha\ket{+}_{c}+\beta\ket{-}_{c}\rightarrow\alpha\ket{++}_{c}+\beta\ket{--}_{c}$~\cite{Horsman2012Surface}. Figure~\ref{fig:3d3dls} shows an example of $X$-type lattice surgery performed on two 3D surface codes.

$Z$-type lattice surgery is analogous to $X$-type lattice surgery. To perform a $Z$-type merge on $\mathcal{SC}_{c}^{(2)}$ and $\mathcal{SC}_{c}^{(2)}$, we first align a $c'$-boundary of one stack with a $c'$-boundary of the other (this aligns the smooth boundaries of the codes). We then add a layer of ancilla qubits (all in the $\ket{+}$ state) and measure new $Z$ stabilizers which join the two lattices. There may also be new $X$ stabilizers and modified $X$ stabilizers at the join. The new $Z$ stabilizers (redundantly) tell us the value of $\overline{Z}_{c}^{(1)}\otimes\overline{Z}_{c}^{(2)}$. The merge implements the mapping $\ket{\psi}_{c}\otimes\ket{\varphi}_{c}\rightarrow a\ket{\psi}_{c}+(-1)^{m}bX\ket{\psi}_{c}$, where $m$ is the outcome of the $\overline{Z}_{c}^{(1)}\otimes\overline{Z}_{c}^{(2)}$ measurement and $\ket{\varphi}_{c}=a\ket{+}+b\ket{-}$~\cite{Horsman2012Surface}. Any $\overline{Z}$-operator of either original code is a valid $\overline{Z}$ operator of the merged code. However, the valid $\overline{X}$ operators of the merged code are membranes of $X$ operators with boundaries which span the $c'$ and $c''$-boundaries of the merged lattice. We can implement a $Z$-type split by measuring a layer of $\mathcal{SC}_{c}$ qubits in the $X$ basis. These measurements implement the following mapping: $\alpha\ket{0}_{c}+\beta\ket{1}_{c}\rightarrow\alpha\ket{00}_{c}+\beta\ket{11}_{c}$~\cite{Horsman2012Surface}. Figure~\ref{fig:3d3dls} shows an example of $Z$-type lattice surgery on two 3D surface codes.

We note that we can simultaneously implement an $X$-type merge on the $\mathcal{SC}_{c}$ codes in different stacks, a $Z$-type merges on the $\mathcal{SC}_{c'}$ codes in different stacks and a $Z$-type merge on the $\mathcal{SC}_{c''}$ codes in different stacks. To do this we prepare a layer of qubits between $c$-boundaries of the two stacks we want to merge. At every vertex in the new layer we place three qubits (one for each pair of codes), prepared in the state $\ket{0}_{c}\ket{+}_{c'}\ket{+}_{c''}$. We then modify the stabilizer groups of all three pairs of codes at once as discussed in the previous paragraphs to merge the three pairs of codes simultaneously. We can also invert this process to do a simultaneous split on all three pairs of codes. 

We illustrate 3D surface code lattice surgery with an example. Consider two $d=3$ rectified cubic lattices placed one above the other as shown in Figure~\ref{fig:3d3dls}. We add a diamond layer (see Section~\ref{subsubsec:lattice structure}) of qubits between the two lattices (vertices of the sublattice with dashed edges in Figure~\ref{fig:3d3dls}). At each vertex we add three qubits (one per code) in the state $\ket{+}_{r}\ket{0}_{g}\ket{+}_{b}$. Next we merge the stabilizer groups of $\mathcal{SC}_{c}^{(1)}$ and $\mathcal{SC}_{c}^{(2)}$, for $c\in\{r,g,b\}$. This implements a $Z$-type merge on $\mathcal{SC}_{r}^{(1)}$ and $\mathcal{SC}_{r}^{(2)}$, an $X$-type merge on $\mathcal{SC}_{g}^{(1)}$ and $\mathcal{SC}_{g}^{(2)}$, and a $Z$-type merge on $\mathcal{SC}_{b}^{(1)}$ and $\mathcal{SC}_{b}^{(2)}$. We now consider each pair of codes with the same colour separately and detail how their stabilizer groups transform. 

\begin{figure}
    \centering
    \includegraphics[width=0.6\columnwidth]{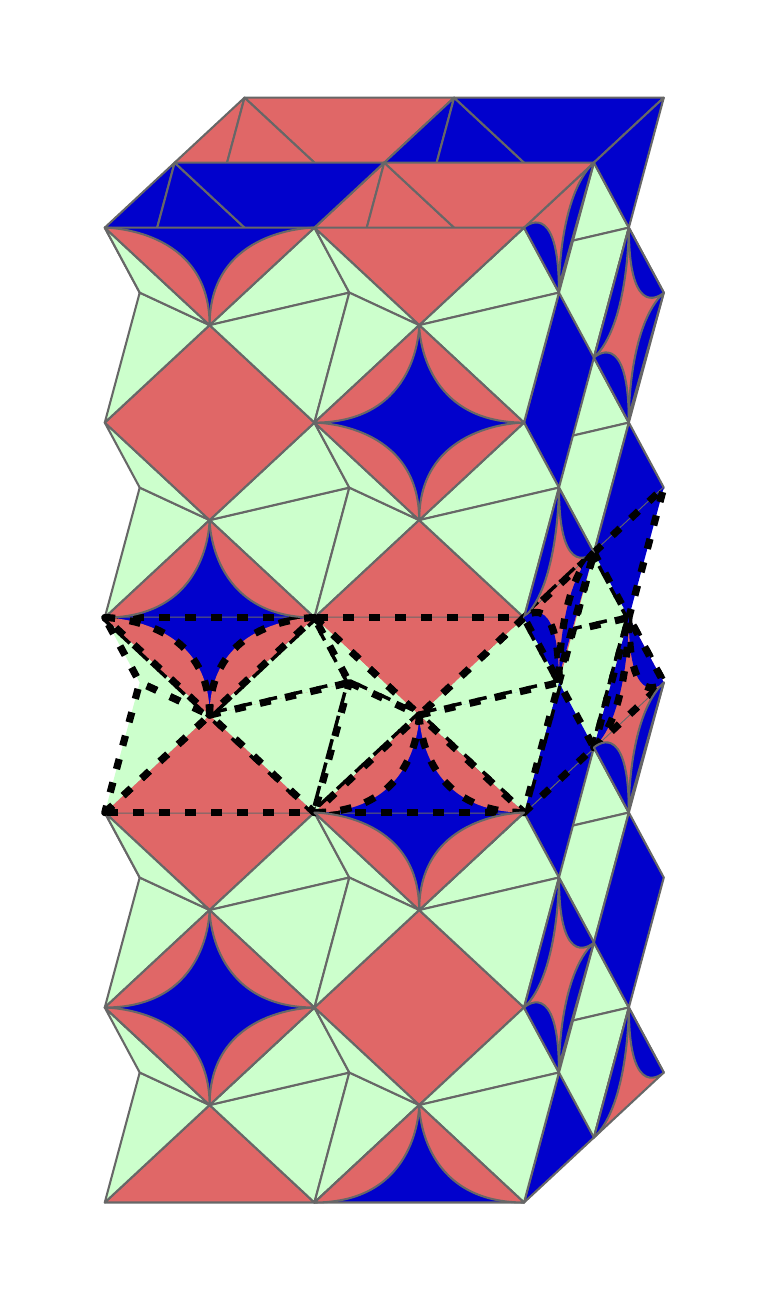}
    \caption{\label{fig:3d3dls}
    Lattice surgery in 3D surface codes. Both initial lattices (sublattices with continuous edges) support three surface codes. We prepare a layer of ancilla qubits (vertices of the sublattice with dashed edges) and measure new stabilizers (faces and cells of the sublattice with dashed edges) to merge codes of the same colour in separate stacks. To undo a merge, we simply measure the layer of ancilla qubits (vertices of the sublattice with dashed edges). In this configuration, we can do $X$-type lattice surgery on $\mathcal{SC}_{g}^{(1)}$ and $\mathcal{SC}_{g}^{(2)}$, $Z$-type lattice surgery on $\mathcal{SC}_{b}^{(1)}$ and $\mathcal{SC}_{b}^{(2)}$ and $Z$-type lattice surgery on $\mathcal{SC}_{r}^{(1)}$ and $\mathcal{SC}_{r}^{(2)}$.
    }
\end{figure}

First of all, consider $\mathcal{SC}_{g}^{(1)}$ and $\mathcal{SC}_{g}^{(2)}$. The code formed by merging these two codes has nine additional $X$ stabilizers (the complete and incomplete octahedra with dashed edges in Figure~\ref{fig:3d3dls}). The merged code also has four additional $Z$ stabilizers (the $rb$-faces parallel to the $g$-boundaries in the sublattice with dashed edges in Figure~\ref{fig:3d3dls}). Some of the $Z$ stabilizers on the boundary are also modified ($rb$-faces in Figure~\ref{fig:3d3dls} with dashed and continuous edges). In total, the merged code has 12 additional physical qubits and 13 additional stabilizer generators. The two original codes each had $n=51$ physical qubits and $n-1$ stabilizer generators so the merged code has $2n+12$ physical qubits and $2(n-1)+13$ stabilizer generators. Hence, the merged code has a single logical qubit, as required. One can also verify that the product of the new $X$ stabilizers is $\overline{X}_{g}^{(1)}\otimes\overline{X}_{g}^{(2)}$.

Next, we consider $\mathcal{SC}_{r}^{(1)}$ and $\mathcal{SC}_{r}^{(2)}$. The code formed by merging these two codes has no additional $X$ stabilizers, but some $X$ stabilizers which were present before the merge are modified ($r$-cuboctahedra and $rb$-faces on the the $b$-boundaries with dashed and continuous edges in Figure~\ref{fig:3d3dls}). The merged code has 16 new $Z$ stabilizers associated with the $gb$-faces of the new cuboctahedra ($gb$-faces in the sublattice with dashed edges in Figure~\ref{fig:3d3dls}). In addition, there are eight new $Z$ stabilizers associated with edges on the $r$-boundaries (blue (dark grey) circular segments with dashed edges in Figure~\ref{fig:3d3dls}). However, these new $Z$ stabilizers are not all independent. In the merged lattice, we have four additional complete cuboctahedra and a single additional complete octahedron when compared with the initial lattices. The stabilizers associated with the $rg$-faces of these polyhedra multiply to the identity, so we must remove a stabilizer from the list of new stabilizer generators for each new complete polyhedron. We also have two additional half octahedra whose edges and faces have associated stabilizers which multiply to the identity (see Figure~\ref{fig:halfoct redund} for an example of such a half octahedron). Therefore, we remove two more stabilizers from the list of new stabilizer generators. Finally, the stabilizers associated with the four edges of $rb$-faces on the $r$-boundaries multiply to the identity, so we must remove half of the new weight two $Z$ stabilizers from the list of new stabilizer generators. The merged code, therefore, has 13 new stabilizer generators and 12 additional qubits. Hence, the merged code encodes a single logical qubit, as required.

The details the $Z$-type lattice surgery on $\mathcal{SC}_{b}^{(1)}$ and $\mathcal{SC}_{b}^{(2)}$ are the same as the details of $Z$-type lattice surgery on $\mathcal{SC}_{r}^{(1)}$ and $\mathcal{SC}_{r}^{(2)}$ (just exchange $r$ and $b$ in the previous paragraph). To verify that the lattice surgery procedures we have described transform two surface codes into a single surface code for any code distance, all we need to do repeat the analysis of Section~\ref{subsubsec:lattice structure} for a slightly different lattice structure. We omit this analysis here as the extension is simple. In Appendix~\ref{app:lattice surgery}, we show another possible arrangement of 3D stacks which allows us to do $X$-type lattice surgery on $\mathcal{SC}_{r}^{(1)}$ and $\mathcal{SC}_{r}^{(2)}$, $Z$-type lattice surgery on $\mathcal{SC}_{g}^{(1)}$ and $\mathcal{SC}_{g}^{(2)}$ and $Z$-type lattice surgery on $\mathcal{SC}_{b}^{(1)}$ and $\mathcal{SC}_{b}^{(2)}$.

\subsection{\label{subsec:2d3dls}2D-3D lattice surgery}

We can do $Z$-type lattice surgery on a 2D surface code and a 3D surface code using procedures which are very similar to 2D surface code lattice surgery. However, performing $X$-type lattice surgery on a 2D surface code and a 3D surface code is more complex. This is because the dimension of the $\overline{Z}$ operators in 2D surface codes and 3D surface codes is the same whereas the dimension of the $\overline{X}$ operators is not. Therefore, we only discuss $Z$-type lattice surgery in this section. We start with a 3D surface code stack and a 2D surface code sheet aligned such that the 2D sheet is in the same plane as the bottom layer of the 3D stack (see Figure~\ref{fig:3d2dls}). To do a lattice surgery merge, we simply measure new $Z$ stabilizers whose product is $\overline{Z}_{2D}\otimes\overline{Z}_{3D}$. The $X$ stabilizers of both codes at the join will also be modified. Figure~\ref{fig:3d2dls} shows an example $Z$-type merge of a 3D code and a 2D code. The effect of the $Z$-type merge on the logical operators is more interesting in the 2D-3D case than the 3D-3D case. The $\overline{Z}$ operators of the original codes are valid $\overline{Z}$ operators of the merged code. However, $\overline{X}$ operators of the merged code are products of membrane operators in the 3D lattice and string operators in the 2D lattice. The merged code is therefore an example of a code with a logical operator which has 2D and 1D parts. We can implement a $Z$-type split by returning to measuring the pre-merge stabilizers.

\begin{figure}
    \centering
    \includegraphics[width=0.9\columnwidth]{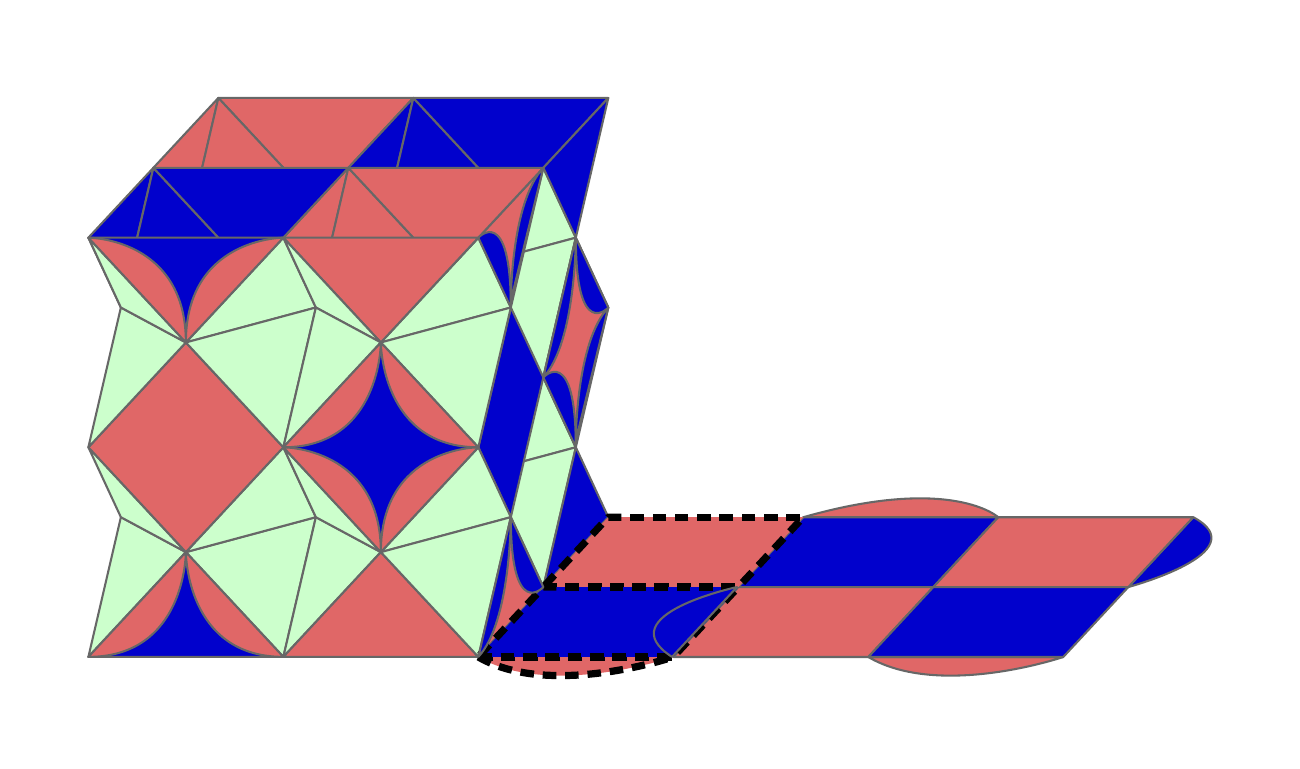}
    \caption{\label{fig:3d2dls} $Z$-type lattice surgery on 3D and 2D surface codes (sublattices with continuous edges). We associate $X$ stabilizers with $b$-faces (dark grey) and $Z$ stabilizers with $r$-faces (medium grey) in the 2D surface code. In the 3D stack we consider $\mathcal{SC}_{b}$ ($X$ stabilizers associated with $b$-cells (dark grey)). The left and right boundaries of the 2D surface code are smooth boundaries and the left and right boundaries of the stack are $r$-boundaries (smooth boundaries in $\mathcal{SC}_{b}$). To implement a lattice surgery merge between the two codes we measure two new $Z$ stabilizers ($r$-faces (medium grey) with dashed edges), whose product is $\overline{Z}_{2D}\otimes\overline{Z}_{3D}$. We also merge the weight two $X$ stabilizer on the left boundary of the 2D code with the weight three $X$ stabilizer associated with the bottom $rb$-face (medium grey) on the right boundary of the 3D code. This stabilizer is represented by the $b$-face (dark grey) with dashed edges in the Figure. To undo the merge operation we return to measuring the pre-merge stabilizers.}
\end{figure}

As we previously stated, we can use lattice surgery to implement $CNOT$ gates and to transfer qubits between different surface codes. Consider the initial state $\ket{\psi}\ket{+}$, where $\ket{\psi}=\alpha\ket{0}+\beta\ket{1}$. Implementing a $Z$-type merge between the two qubits followed by a $Z$-type split produces the state $\alpha\ket{00}+\beta\ket{11}$. If we measure the first qubit in the $X$ basis, $\ket{\psi}$ is transferred to the second qubit (up to a $Z$ correction). The lattice surgery $CNOT$ procedure is similar to the procedure we have just described. Consider the state $\ket{\psi}\ket{+}\ket{\phi}$. To perform a $CNOT$ with $\ket{\psi}$ as the control and $\ket{\phi}$ as the target we first do a $Z$-type merge of $\ket{\psi}$ and $\ket{+}$ followed by a $Z$-type split. The second step is to do an $X$-type merge of $\ket{\phi}$ and $\ket{+}$ followed by an $X$-type split. There are also some single qubit corrections that may be necessary which we have omitted. For the full details of this $CNOT$ procedure see~\cite{Horsman2012Surface}. We can also use a chain of lattice surgery operations to perform a multi-target $CNOT$ gate as shown in~\cite{Herr2017Lattice}.

We emphasize that the lattice surgery procedure we have explained in this Section is one of many code deformation procedures which we could use to transfer information from a 3D surface code to a 2D surface code. For example, in Appendix~\ref{app:lattice surgery}, we give a different implementation of 2D-3D surface code lattice surgery. It is also possible to transfer information using a `code switching' deformation (in the spirit of~\cite{Bombin2016Dimensional}), where we transform a 3D surface code into a 2D surface code by measuring all but one layer of physical qubits in the $X$ basis. Finally, we note that the $Z$-type lattice surgery operations we have described can also be used to do $Z$-type lattice surgery between two 3D surface codes.

\section{\label{sec:arch} 3D Surface Code Architectures}

In this section, we propose two universal quantum computing architectures which use 3D surface codes. But first we discuss decoding 3D surface codes. 

\subsection{\label{subsec:dec} Decoding 3D surface codes}

Estimating the error thresholds of 3D surface codes is beyond the scope of this article. Instead we discuss possible decoding strategies for 3D surface codes and reason about the error thresholds we might expect. The 3D surface code is interesting from a decoding point of view because of the asymmetry between membrane-like $X$ errors and string-like $Z$ errors. This asymmetry means that different decoding strategies may be needed for $X$ and $Z$ errors.  

We can upper bound the error thresholds of topological codes by relating the codes to condensed matter models~\cite{Dennis2002Topological}. The phase diagram of the condensed matter model will then give us an estimate of the optimal error threshold of the code. Using this technique, the optimal error threshold of 2D surface codes has been estimated to be $\approx 11\%$~\cite{Honecker2001Universality}, for a stochastic noise model where $X$ and $Z$ errors happen independently with probability $p$ and measurements are perfect. For the 3D surface code, the optimal error threshold for the same noise model is $\approx 3.3\%$. We can break this error threshold down further: for a noise model where $Z$ ($X$) errors happen with probability $p$ and measurements are perfect the error threshold is $p_{th}^{Z}\approx 3.3\%$~\cite{Ohno2004Phase} ($p_{th}^{X}\approx 23.5\%$~\cite{Ozeki1998Multicritical,Hasenbusch2007Magnetic}). 

The above error thresholds will not be achievable in practice due to measurement errors. For the 2D surface code, simulations of the full syndrome extraction circuits indicate an error threshold between $0.5\%$ and $1.1\%$ (see~\cite{Stephens2014Fault} and references therein). To the best of our knowledge, no similar simulation has been performed for 3D surface codes. However, we anticipate that the 3D surface code error threshold will be lower than the corresponding 2D surface code error threshold because of the higher dimensionality of the lattice and the larger weight stabilizers in 3D surface codes. 

The most popular 2D surface code decoder is a minimum-weight perfect-matching (MWPM) algorithm~\cite{Edmonds1965Paths,Kolmogorov2009Blossom,Fowler2012Surface}. We could use MWPM to decode $Z$ errors in cubic surface codes and tetrahedral-octahedral surface codes. Alternatively, we could use the recently proposed Union-Find decoder~\cite{Delfosse2017UnionFind}, which has slightly worse performance than MWPM, but a much faster runtime. There are a number of approaches we could take to decoding membrane-like $X$ errors in 3D surface codes. Duivenvoorden \emph{et al}.\ estimated the $X$ error threshold of 3D cubic surface codes using an efficient renormalization decoder~\cite{Duivenvoorden2017Renormalization}. They found an $X$ error threshold of $p_{th}^{X}=17.2\pm 1\%$ for an error model with perfect measurements. It would be interesting to generalize Duivenvoorden \emph{et al}.'s decoder to non-cubic surface codes such as tetrahedral-octahedral surface codes. Another option for decoding membrane-like $X$ errors in 3D surface codes is to use a generalisation of Toom's rule as the decoding algorithm~\cite{Toom1980Stable,Dennis2002Topological,Kubica2018Cellular,Kubica2018Thesis}. For such a decoder, Kubica estimated an $X$ error threshold of $p_{th}^{X}\approx 2\%$ for 3D surface codes with periodic boundaries (3D toric codes)~\cite{Kubica2018Thesis}. This error threshold is for a noise model where $X$ errors and measurement errors both occur with probability $p$. In future work, we intend to extend this result to 3D surface codes with boundaries.

\subsection{\label{subsec:hybrid} Hybrid 2D-3D surface-code architecture}

In this section, we present a hybrid 2D-3D surface code architecture based on~\cite{Horsman2012Surface}. In our hybrid architecture, the main component is a sheet of 2D surface code patches. Lattice surgery allows us to do $CNOT$ gates between different patches. We can also do Hadamard gates easily as explained in~\cite{Horsman2012Surface}. We use 3D surface codes as $CCZ$ state ($\ket{CCZ}=CCZ\ket{+++}$) factories in our hybrid architecture, replacing the magic state distillation used in the original architecture. We can fault-tolerantly create $CCZ$ states in a 3D surface code stack as long as we have a decoder with an error threshold. We use $Z$-type lattice surgery to transfer $CCZ$ states from a stack of 3D surface codes into the sheet of 2D surface codes. There is some subtlety involved in transferring three logical qubits from a single 3D surface code stack to a 2D surface code sheet so we describe this procedure now. We consider a 3D surface code stack which can interface with a single 2D surface code. This means that we can only transfer encoded states from one of the three 3D surface codes (say $\mathcal{SC}_{r}$) to the 2D surface code. Other configurations are possible, but we will concentrate on this most basic configuration. We refer to the logical qubit encoded in $\mathcal{SC}_{c}$ as the $c$-qubit. Assume that we have prepared $CCZ$ state in the 3D surface code stack. We can transfer the $r$-qubit to the 2D surface code easily using $Z$-type lattice surgery (see Section~\ref{subsec:2d3dls}). Next we want to transfer the $g$-qubit. We must transfer the state of the $g$-qubit in the stack to the $r$-qubit first. However, we need two qubits in the stack to be ancillas in order to do this (see Equation~\ref{eq:H and tele}), and only one is available. Instead, we transfer the state of the $g$-qubit to the $r$-qubit, with a $H$ gate applied (see Equation~\ref{eq:H and tele}). Next, we transfer this state to the 2D surface code where we can undo the $H$ gate. Finally, we transfer the state of the $b$-qubit to the $r$-qubit (we now have enough ancillas) and transfer this state to the 2D surface code. 

Once we have an encoded $CCZ$ state in our sheet of 2D surface codes, we can implement a $CCZ$ gate on any three qubits using a state injection circuit and some SWAP gates. Figure~\ref{fig:ccz inject} shows a state injection circuit containing Pauli, $H$ and $CNOT$ gates that uses one $CCZ$ state to implement a $CCZ$ gate. We constructed this circuit using the methodology described in~\cite{Zhou2000Methodology}. To summarize, we have explained how to implement the universal gate set $\{X, Z, H, CNOT, CCZ\}$ in our hybrid 2D-3D surface code architecture. 

\begin{figure}[ht]
    \centering
    \includegraphics[width=\linewidth]{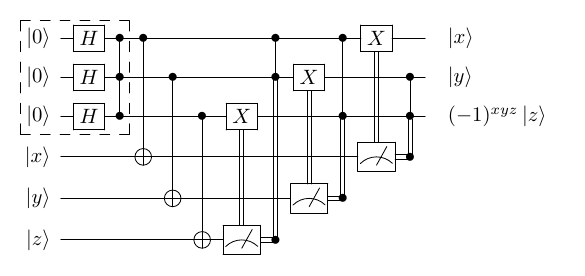}
    \caption{\label{fig:ccz inject} A circuit that consumes one $CCZ$ state (dashed box) to implement a $CCZ$ gate on the bottom three qubits. We note that $H_{t}\cdot CNOT_{ct}\cdot H_{t}=CZ$, where $c$ and $t$ refer to the control and target qubits.}
\end{figure}

\subsection{\label{app:3D arch} 3D surface-code architecture}

In this section, we present a quantum computing architecture where every qubit is encoded in a 3D surface code. We consider a large rectified cubic lattice with 3D `patches' each containing three logical qubits. Each patch is a distance $d$ rectified cubic lattice adjacent to six identical patches. We can do lattice surgery on adjacent patches as described in Section~\ref{subsec:3d3dls}. We adopt a Euclidean coordinate system and associate each of the axes with a particular colour. For example, we associate the $x$-direction with $r$ which implies that we can do $X$-type lattice surgery on $r$-qubits (qubits encoded in $\mathcal{SC}_{r}$ codes) in patches which are adjacent in the $x$-direction. Similarly, we can do $X$-type lattice surgery on $g$-qubits ($b$-qubits) which are adjacent in the $y$-direction ($z$-direction). This means that we can transfer a qubit from one patch to any of its adjacent patches using $X$-type or $Z$-type lattice surgery. 

In our architecture we use half of the patches in the lattice as `data patches' and half as `ancilla patches'. Data patches contain three logical data qubits and ancilla patches contain three logical ancilla qubits. We can do $CNOT$ gates between any two qubits in data patches which are adjacent to the same ancilla patch using lattice surgery. If the two data qubits have different colours then we need to use two logical qubits in the ancilla patch during the procedure. For example, imagine we want to do a $CNOT$ between the $r$-qubit (control) and $g$-qubit (target) in the same data patch. First of all, we do a $Z$-type merge of the $r$-qubit in the data patch and the $r$-qubit in an adjacent ancilla patch. We then undo this merge with a $Z$-type split. Next, we transfer the state of the $r$-qubit in the ancilla patch to the $g$-qubit in the same ancilla patch, using the procedure in Equation~\ref{eq:H and tele}. The next step is to do an $X$-type merge of the $g$-qubits in the data patch and the ancilla patch. Finally, we undo this merge with an $X$-type split and apply some Pauli corrections. The procedure we have just described implements a $CNOT$ gate between the $r$-qubit and $g$-qubit in the same data patch. 

$CNOT$ gates allow us to swap any two data qubits in data patches which are adjacent to the same ancilla patch. As we have previously shown, we can transversally implement $CZ$ and $CCZ$ in a single data patch. Finally, we can do a $H$ gate on a single qubit in a data patch by the following method. We first transfer the qubit to an adjacent ancilla patch using lattice surgery. Next we do a single qubit $H$ using the procedure in Equation~\ref{eq:H and tele} before transferring the qubit back to its original data patch. In the architecture we have just described we can swap arbitrary data qubits and implement a universal gate set in each data patch. $CCZ$ gates can be performed in parallel on all data qubits, $CZ$ gates can be performed in parallel on two thirds of the data qubits and $H$ gates can be performed in parallel on a third of the data qubits. This architecture requires no magic state distillation or state injection.

\section{\label{sec:discuss} Discussion}

In this article, we introduced the rectified picture of 3D surface codes. We used the rectified picture to analyse stacks of three 3D surface codes, showing that $CCZ$ is transversal in these codes. In addition, we detailed 3D surface code architectures which allow us to do universal quantum computing without magic state distillation.

As we mentioned in Section~\ref{sec:intro}, the large resource cost of magic state distillation has motivated research into alternative implementations of non-Clifford gates in topological codes. To reason about the resource scaling of different architectures we use a spacetime overhead metric. Roughly speaking, an architecture which requires $n$ physical qubits and $d$ rounds of syndrome extraction per operation has a spacetime overhead of $nd$. 2D surface code architectures and 3D gauge color code architectures have a similar spacetime overhead scaling. Distance $d$ 2D surface codes have $O(d^{2})$ physical qubits and require $O(d)$ rounds of syndrome extraction to cope with measurement errors. Distance $d$ 3D gauge color codes have $O(d^{3})$ physical qubits but only require $O(1)$ rounds of syndrome extraction. The structure of the error syndrome gives us information which we can use to diagnose measurement errors immediately. That is, 3D gauge color codes can be decoded in a single-shot fashion~\cite{Bombin2015GaugeCC,Bombin2015Single}. If we want to assess the resource scaling of our 3D surface code architectures compared with magic state distillation architectures, we need to understand 3D surface code decoding in more detail. A distance $d$ 3D surface code requires $O(d^{3})$ physical qubits but an unknown number of rounds of syndrome extraction. Membrane-like $X$ errors in 3D surface codes without boundaries can be decoded using a single-shot cellular automaton decoder~\cite{Kubica2018Thesis,Kubica2018Thesis}. However, it seems unlikely that we will be able to use a single-shot decoder to decode string-like $Z$ errors in 3D surface codes. Nevertheless, due to the links between surface codes and color codes, it may be possible to construct a `3D gauge surface code' where single-shot error correction is possible for both $X$ and $Z$ errors. 

It will also be important to estimate the numerical value of the error threshold for both cubic surface codes and tetrahedral-octahedral surface codes. This is because the resources required in a particular architecture depend strongly on the value of the error threshold. The error threshold of the gauge color code has been estimated to be $\approx 0.31\%$~\cite{Brown2016Gauge}, for an error model where qubit errors and measurement errors occur with the same probability. We would expect to observe a smaller error threshold if we were to simulate the full syndrome extraction circuits. Therefore, even with the similar resource scaling, we anticipate that 3D gauge color code architectures would require more physical qubits than 2D surface code architectures which use magic state distillation (with current qubit technologies). However, we should note that much more work has gone into optimizing 2D surface code architectures than gauge color code architectures, so an error threshold of $p_{th}\approx 0.31\%$ for gauge color codes may be pessimistic. In future work, we plan to investigate decoding 3D surface codes on both cubic and tetrahedral-octahedral lattices. Once we have estimates of the error thresholds we will be able to definitively compare the resources required by 3D surface code architectures and magic state distillation architectures. It is also interesting to consider an alternative architecture where 3D surface codes and magic state distillation are combined. For example, we could use 3D surface codes to prepare reasonably high fidelity $CCZ$ states which we would then feed in to a magic state distillation protocol (\emph{e.g}.~\cite{Paetznick2013Universal}). This would remove the need for multiple rounds of magic state distillation and could therefore lead to reduced resource overheads in some scenarios. There are also alternative 3D surface code architectures we could consider. One of the most popular approaches to 2D surface code quantum computing is to encode logical qubits as pairs of defects (stabilizers which have been turned off) and braid defects to perform logical gates~\cite{Fowler2012Surface}. It should be possible to generalize this approach to 3D surface codes. 

It seems likely that both the rectified picture and code concatenation transformations could be generalized to higher dimensional ($D\geq 4$) surface codes. These generalisations could give us some insight into the structure and transversal gates of higher dimensional surface codes. Most importantly, the question of whether magic state distillation or transversal gates in 3D topological codes is the best method for promoting 2D topological code architectures to universality remains open. We hope that our work contributes towards answering this question. 

\begin{acknowledgments}
    The authors would like to thank Hussain Anwar, Earl Campbell, Alex Kubica and Paul Webster for helpful discussions. We thank the anonymous referees for helpful comments, and for pointing out a simpler proof of the transversality of $CZ$ in stacked 3D surface codes. MV is supported by the EPSRC (grant number EP/L015242/1).
\end{acknowledgments}

\appendix

\section{\label{app:d2 code} Explicit construction of the $d=2$ 3D surface code stack}

Here, we detail list the stabilizer generators and logical operators of three surface codes in the $d=2$ rectified cubic stack. Figure~\ref{fig:d2 cube} shows the $d=2$ rectified cubic lattice. Each code in the stack is a [[12,1,2]] 3D surface code. We label the physical qubits in each code as shown in Figure~\ref{fig:d2 cube}. $P_{i}$ denotes a Pauli operator acting on qubit $i$. 

The stabilizer generators of $\mathcal{SC}_{r}$ are: 
\begin{equation}
    \begin{gathered}
        X_{5}X_{6}X_{7}X_{8}X_{9}X_{10}X_{11}X_{12}, \\
        X_{1}X_{3}X_{5}, 
        X_{2}X_{4}X_{7}, \\
        Z_{6}Z_{9}, 
        Z_{6}Z_{10}, \\
        Z_{8}Z_{11}, 
        Z_{8}Z_{12}, \\
        Z_{1}Z_{5}Z_{6}, 
        Z_{2}Z_{6}Z_{7}, \\
        Z_{4}Z_{7}Z_{8}, 
        Z_{3}Z_{5}Z_{8}.
    \end{gathered}
    \label{eq:scr d2 stabs}
\end{equation}
Example $\mathcal{SC}_{r}$ logical operators are $\overline{Z}_{r}=Z_{1}Z_{3}$ and $\overline{X}_{r}=X_{3}X_{4}X_{8}X_{11}X_{12}$. 

The stabilizer generators of $\mathcal{SC}_{g}$ are: 
\begin{equation}
    \begin{gathered}
        X_{1}X_{5}X_{6}X_{9}, 
        X_{2}X_{6}X_{7}X_{10}, \\
        X_{3}X_{5}X_{8}X_{11}, 
        X_{4}X_{7}X_{8}X_{12}, \\
        Z_{1}Z_{3}Z_{5}, 
        Z_{1}Z_{2}Z_{6}, \\
        Z_{2}Z_{4}Z_{7}, 
        Z_{3}Z_{4}Z_{8}, \\
        Z_{5}Z_{9}Z_{11}, 
        Z_{6}Z_{9}Z_{10}, \\
        Z_{7}Z_{10}Z_{12}.
    \end{gathered}
    \label{eq:scg d2 stabs}
\end{equation}
Example $\mathcal{SC}_{g}$ logical operators are $\overline{Z}_{g}=Z_{1}Z_{9}$ and $\overline{X}_{g}=X_{1}X_{2}X_{3}X_{4}$. 

The stabilizer generators of $\mathcal{SC}_{b}$ are: 
\begin{equation}
    \begin{gathered}
        X_{1}X_{2}X_{3}X_{4}X_{5}X_{6}X_{7}X_{8}, \\
        X_{6}X_{9}X_{10}, 
        X_{8}X_{11}X_{12}, \\
        Z_{1}Z_{5}, 
        Z_{3}Z_{5}, \\
        Z_{2}Z_{7}, 
        Z_{4}Z_{7}, \\
        Z_{5}Z_{8}Z_{11}, 
        Z_{5}Z_{6}Z_{9}, \\
        Z_{6}Z_{7}Z_{10}, 
        Z_{7}Z_{8}Z_{12}.
    \end{gathered}
    \label{eq:scb d2 stabs}
\end{equation}
Example $\mathcal{SC}_{b}$ logical operators are $\overline{Z}_{b}=Z_{1}Z_{2}$ and $\overline{X}_{b}=X_{1}X_{3}X_{5}X_{9}X_{11}$. 

\begin{figure}
    \includegraphics[width=0.8\columnwidth]{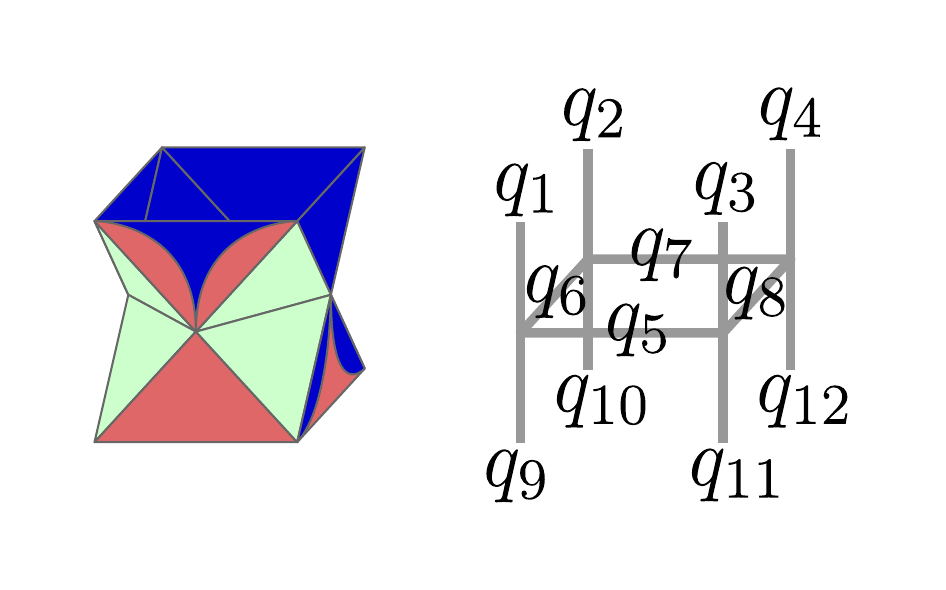}
    \caption{\label{fig:d2 cube} The $d=2$ rectified cubic lattice (left) and the $\mathcal{SC}_{g}$ primal lattice in the Kitaev picture (right), with qubit labels. In the Kitaev picture, qubits are placed on primal lattice edges, $X$ stabilizers are associated with primal lattice vertices and $Z$ stabilizers are associated with primal lattice faces. $\mathcal{SC}_{r}$ and $\mathcal{SC}_{b}$ also have physical qubits at the same locations as the labelled $\mathcal{SC}_{g}$ physical qubits. Therefore, we use the same label to refer to qubits in different codes that occupy the same position.}
\end{figure}

\section{\label{app:parallelepiped} Alternative rectified cubic lattice}

We can construct an alternative family of stacked 3D surface codes by choosing lattices with different boundaries. The lattices in this family have the global structure of parallelepipeds so we refer to them as parallelepiped lattices. Parallelepiped lattices have two boundaries which slice layers of cuboctahedra in half, like the lattices we discussed in the main text. However, parallelepiped lattices do not have boundaries which slice between layers of cuboctahedra. Instead they have boundaries which slice one colour of cuboctahedra in half and leave the other colour of cuboctahedra intact. Figure~\ref{fig:para} shows a $d=3$ parallelepiped lattice. We define three surface codes on parallelepiped lattices by associating $X$ stabilizers with $c$-cells and $Z$ stabilizers with $c'c''$-faces. We must also add some stabilizers on the boundaries to ensure that the boundaries have the correct properties (red (medium grey) and blue (dark grey) circular segments in Figure~\ref{fig:para}). The family of surface codes defined on parallelepiped lattices has the same distance as the family of codes we discussed in the main text. However, the parallelepiped lattices have more physical qubits per logical qubit and are more complex to tessellate. 

\begin{figure}
    \includegraphics[width=0.7\columnwidth]{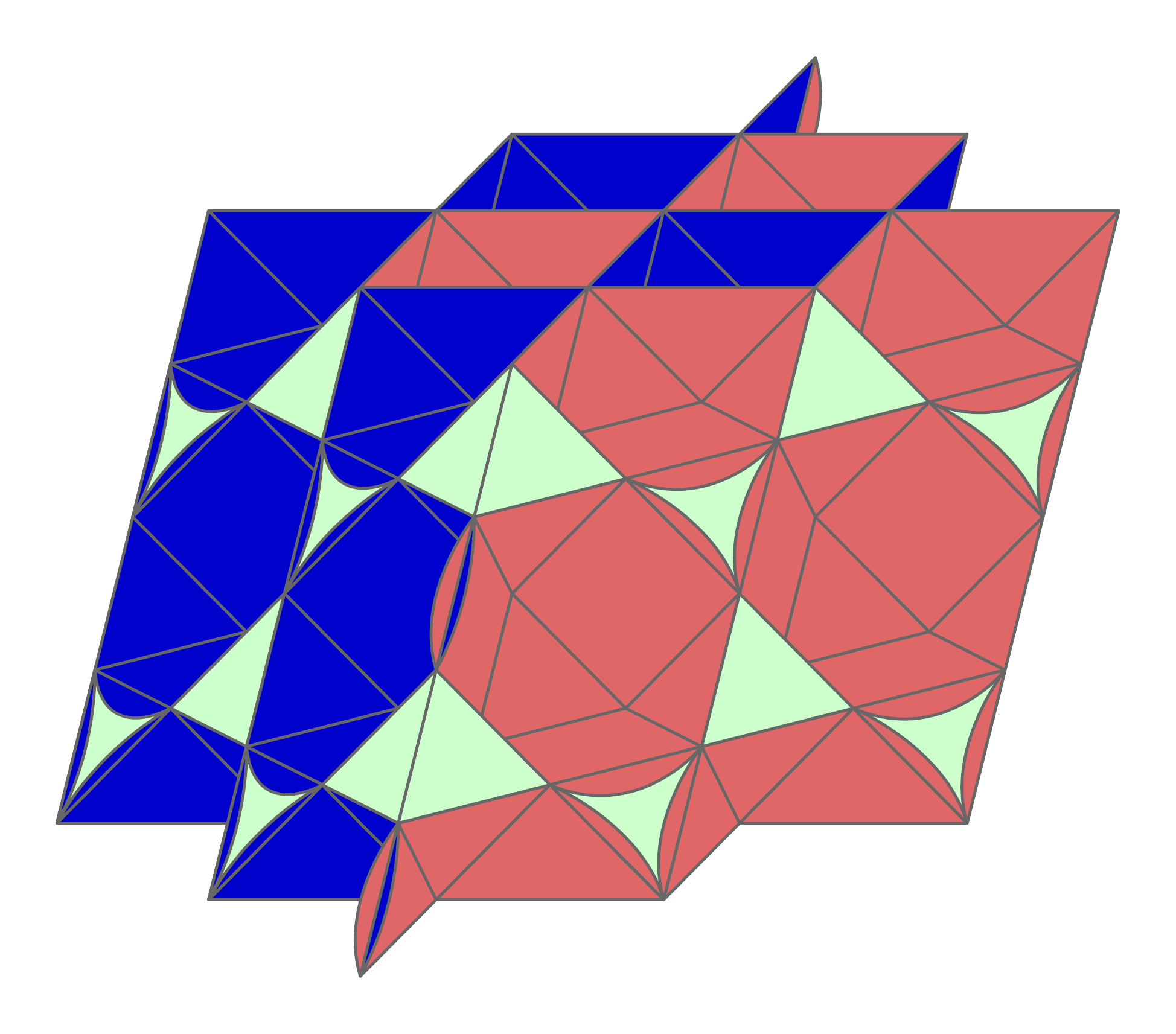}
    \caption{\label{fig:para} A lattice from an alternative family of rectified cubic lattices which can support three 3D surface codes. The top and bottom boundaries are the same type as the top and bottom boundaries in Figure~\ref{fig:extra stabs}. The right boundary slices $r$-cuboctahedra (medium grey) in half and leaves $b$-cuboctahedra (dark grey) intact. The left boundary slices $b$-cuboctahedra (dark grey) in half and leaves $r$-cuboctahedra (medium grey) intact. 
    }
\end{figure}

\section{\label{app:lemma} Proof of Lemma~\ref{lem:parity}}

Lemma~\ref{lem:parity} was required in the proof of the transversality of $CCZ$ for stacked 3D surface codes. We restate it here:
\begin{customlemma}{2}
    Given a finite set of $k$ binary vectors $\{a_{j}\}$ with the same length, the parity of their sum is equal to the sum of their parities. 
    \label{lem:parity app}
\end{customlemma}

\begin{proof}
    An equivalent statement of the Lemma~\ref{lem:parity} is
    \begin{equation}
        \sum_{j=1}^{k}|a_{j}|-|\sum_{j=1}^{k}a_{j}|=2t, 
        \label{eq:lemma}
    \end{equation}
    where $t$ is a positive integer and $|a_{j}|$ denotes the Hamming weight of $a_{j}$. We prove Lemma~\ref{lem:parity} by induction. Consider the $k=2$ case. We have 
    \begin{equation}
        |a_{1}+a_{2}|=|a_{1}|+|a_{2}|-2\mathcal{O}(a_{1},a_{2}),
        \label{eq:base case}
    \end{equation}
    where $\mathcal{O}(a,b)$ is the overlap of $a$ and $b$ \emph{i.e}.\ the number of positions where both $a$ and $b$ are equal to one. Rearranging, we have:
    \begin{equation}
        |a_{1}|+|a_{2}|-|a_{1}+a_{2}|=2\mathcal{O}(a_{1},a_{2}).
        \label{eq:base case rearrange}
    \end{equation}
    Now assume Equation~\ref{eq:lemma} is valid for the $k=n$ case. Consider the $k=n+1$ case:
    \begin{equation}
        \begin{split}
        |\sum_{j=1}^{n+1}a_{j}|&=|\sum_{j=1}^{n}a_{j}|+|a_{n+1}|-2\mathcal{O}\left(\sum_{j=1}^{n}a_{j},a_{n+1}\right) \\
        &=\sum_{j=1}^{n}|a_{j}|-2t+|a_{n+1}|-2\mathcal{O}\left(\sum_{j=1}^{n}a_{j},a_{n+1}\right) \\
        &=\sum_{j=1}^{n+1}|a_{j}|-2\left(t+\mathcal{O}\left(\sum_{j=1}^{n}a_{j},a_{n+1}\right)\right).
        \end{split}
        \label{eq:n+1 case}
    \end{equation}
\end{proof}

\section{\label{app:2D cz} Transversal CZ in 2D surface codes}

In this appendix we show that $CZ$ is transversal for 2D surface codes. Consider a 2D surface code lattice in the rotated picture with faces coloured $r$ and $b$ (\emph{e.g}.\ Figure~\ref{fig:2dsc}). We define a stack of two 2D surface codes on the same lattice. Similarly to the rectified picture of 3D surface codes, we place two physical qubits at each vertex of the lattice (one per code). In the first surface code, $\mathcal{SC}_{1}$, we associate $X$ stabilizers with $r$-faces and $Z$ stabilizers with $b$-faces. In the second surface code, $\mathcal{SC}_{2}$, we associate $X$ stabilizers with $b$-faces and $Z$ stabilizers with $r$-faces. To show that $CZ$ is transversal for this stack of codes, we need to find a transversal operator that implements the following mapping at the logical level:
\begin{equation}
    \begin{split}
        I_{1}I_{2}&\xrightarrow{CZ}I_{1}I_{2}, \\
        Z_{j}&\xrightarrow{CZ}Z_{j}, \\
        X_{j}&\xrightarrow{CZ}X_{j}Z_{k}\quad j\neq k,
    \end{split}
    \label{eq:CZ}
\end{equation}
where $j,k\in\{1,2\}$. 

Consider the action of the transversal operator $\overline{CZ}=CZ^{\otimes n}$, where $n$ is the number of vertices in the lattice and the $CZ$ gates act on the pairs of qubits at each vertex. Our $\overline{CZ}$ operator will leave logical $\overline{Z}_{j}$ operators invariant as these operators consist entirely of $Z$ operators. $\overline{CZ}$ will map logical $\overline{X}_{j}$ operators to $\overline{X}_{j}\overline{Z}_{k}$ because $\overline{X}_{j}$ operators consist entirely of $X$ operators and $\overline{X}_{j}$ operators in one code have the same support as $\overline{Z}_{k}$ operators in the other code. $\overline{CZ}$ has no effect on the $Z$ stabilizers of either code but it maps $X$ stabilizers in one code to a tensor product of the original $X$ stabilizer and a $Z$ stabilizer in the other code. To see this, consider an $X$ stabilizer associated with a $r$-face $f_{r}$ in $\mathcal{SC}_{1}$. Under the action of $\overline{CZ}$, this operator is mapped the tensor product of itself and a product of $Z$ operators acting on the $\mathcal{SC}_{2}$ qubits at the vertices of $f_{r}$. This is nothing more than a $\mathcal{SC}_{2}$ $Z$ stabilizer. Therefore, $\overline{CZ}$ maps the logical identity to the logical identity. We have shown that $\overline{CZ}=CZ^{\otimes n}$ acts as a logical $\overline{CZ}$, implementing the mapping described in Equation~\ref{eq:CZ} at the logical level.

\begin{figure}
    \includegraphics[width=0.9\columnwidth]{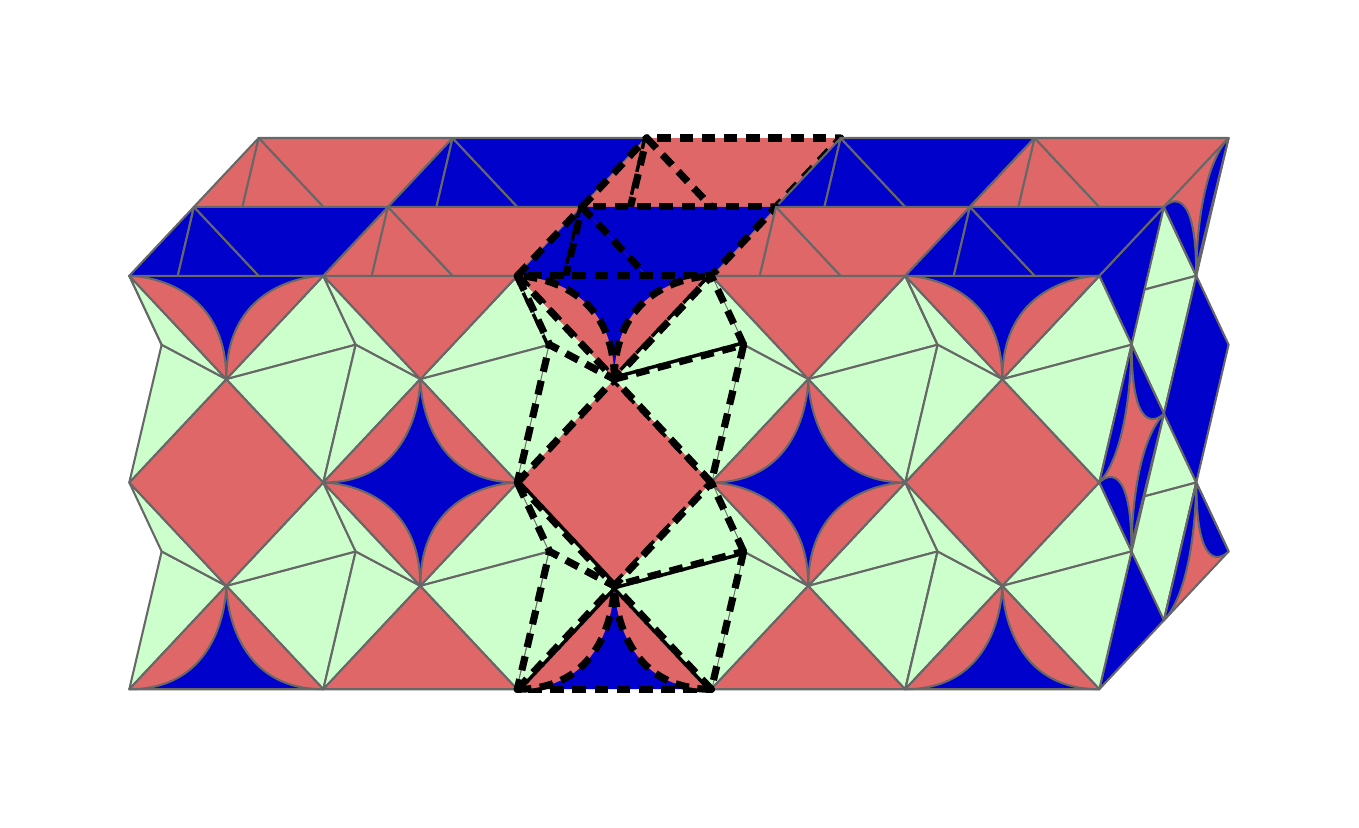}
    \caption{\label{fig:app3d3dls} 
    This configuration of lattices allows us to do $X$-type lattice surgery on $\mathcal{SC}_{r}^{(1)}$ and $\mathcal{SC}_{r}^{(2)}$, $Z$-type lattice surgery on $\mathcal{SC}_{g}^{(1)}$ and $\mathcal{SC}_{g}^{(2)}$ and $Z$-type lattice surgery on $\mathcal{SC}_{b}^{(1)}$ and $\mathcal{SC}_{b}^{(2)}$. Six ancilla qubits are required to do the lattice surgery merges for each pair of codes. The additional (and modified) stabilizers present in the merged codes are associated with the elements of the sublattice with dashed edges. 
    }
\end{figure}

\section{\label{app:lattice surgery} Additional lattice surgery examples}

In this appendix, we give further examples of 3D surface code lattice surgery. First, we consider two 3D surface code stacks, which we denote as $\mathcal{SC}_{c}^{(i)}$, where $c\in\{r,g,b\}$ denotes the colour of the $X$ stabilizers and $i\in\{1,2\}$ indexes the stack. Figure~\ref{fig:app3d3dls} shows a configuration which allows us to do $X$-type lattice surgery on $\mathcal{SC}_{r}^{(1)}$ and $\mathcal{SC}_{r}^{(2)}$. With this lattice configuration we can also do $Z$-type lattice surgery on $\mathcal{SC}_{b}^{(1)}$ and $\mathcal{SC}_{b}^{(2)}$, and $Z$-type lattice surgery on $\mathcal{SC}_{g}^{(1)}$ and $\mathcal{SC}_{g}^{(2)}$. Figure~\ref{fig:app3d2dls} shows a configuration of lattices which allows us to do lattice surgery on a 3D surface code and a 2D surface code. Unlike the lattice surgery example shown in Figure~\ref{fig:3d2dls}, this configuration requires ancilla qubits and hence is less efficient. 

\begin{figure}
    \includegraphics[width=0.9\columnwidth]{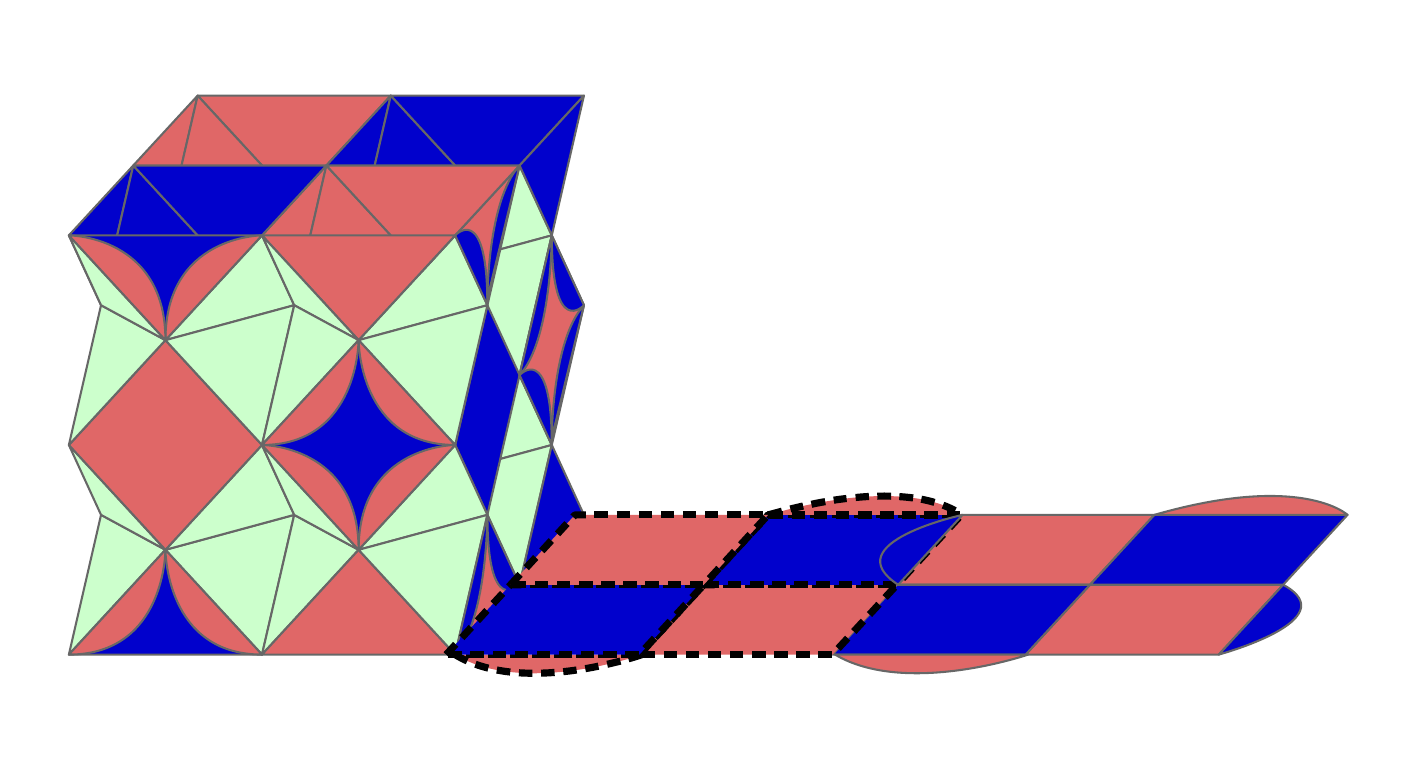}
    \caption{\label{fig:app3d2dls}
    $Z$-type lattice surgery on a 3D surface code and a 2D surface code (sublattices with continuous edges). We consider $\mathcal{SC}_{b}$ in the stack and we associate $X$ stabilizers with $b$-faces (dark grey) and $Z$ stabilizers with $r$-faces (medium grey) in the 2D surface code. Three ancilla qubits are required to do the lattice surgery merge. Four additional $Z$ stabilizers are present in the merged code ($r$-faces (medium grey) with dashed edges) and two $X$ stabilizers from the original codes are modified in the merged code ($b$-faces (dark grey) with dashed edges).
    }
\end{figure}

\end{document}